\documentclass[11pt,letterpaper]{article}
\usepackage[margin=1in]{geometry}
\usepackage{changepage}  
\usepackage[T1]{fontenc}
\usepackage{newpxtext}
\usepackage{amsmath,amsthm,amssymb,mathtools,aliascnt,newpxmath,bm}
\usepackage{microtype,booktabs,array,tabularx}
\usepackage[font=small,labelfont=bf]{caption}
\usepackage{algorithm,algpseudocode,needspace,etoolbox,placeins,flafter}
\usepackage[authoryear,round]{natbib}
\setcitestyle{citesep={,}}
\usepackage{xcolor,xurl}
\usepackage{graphicx,tikz}
\usepackage{todonotes}
\usetikzlibrary{arrows.meta,positioning,calc}
\definecolor{linkblue}{RGB}{25,54,110}
\usepackage[colorlinks=true,linkcolor=linkblue,citecolor=linkblue,urlcolor=linkblue,
 pdftitle={Constant Swap Regret in General-Sum Games via Two-Scale Higher-Order Optimism},
 pdfsubject={Blum--Mansour dynamics with prediction at two scales},
 pdfkeywords={swap regret, optimistic learning, correlated equilibrium, stationary distributions}]{hyperref}
\usepackage[nameinlink,noabbrev]{cleveref}
\makeatletter
\providecommand*{\theHALG@line}{}
\renewcommand*{\theHALG@line}{\thealgorithm.\arabic{ALG@line}}
\makeatother

\numberwithin{equation}{section}

\newtheorem{theorem}{Theorem}[section]
\newaliascnt{lemma}{theorem}
\newtheorem{lemma}[lemma]{Lemma}
\aliascntresetthe{lemma}
\newaliascnt{proposition}{theorem}

\aliascntresetthe{proposition}
\newaliascnt{corollary}{theorem}

\aliascntresetthe{corollary}
\newaliascnt{remark}{theorem}
\theoremstyle{remark}
\newtheorem{remark}[remark]{Remark}
\aliascntresetthe{remark}
\theoremstyle{plain}
\crefname{theorem}{Theorem}{Theorems}
\crefname{lemma}{Lemma}{Lemmas}
\crefname{proposition}{Proposition}{Propositions}
\crefname{corollary}{Corollary}{Corollaries}
\crefname{remark}{Remark}{Remarks}
\crefname{section}{Section}{Sections}
\crefname{algorithm}{Algorithm}{Algorithms}
\Crefname{theorem}{Theorem}{Theorems}
\Crefname{lemma}{Lemma}{Lemmas}
\Crefname{proposition}{Proposition}{Propositions}
\Crefname{corollary}{Corollary}{Corollaries}
\Crefname{remark}{Remark}{Remarks}
\Crefname{section}{Section}{Sections}
\Crefname{algorithm}{Algorithm}{Algorithms}
\AtBeginEnvironment{theorem}{\Needspace{5\baselineskip}}
\AtBeginEnvironment{lemma}{\Needspace{5\baselineskip}}
\AtBeginEnvironment{proposition}{\Needspace{5\baselineskip}}
\AtBeginEnvironment{remark}{\Needspace{4\baselineskip}}
\pretocmd{\section}{\Needspace{8\baselineskip}}{}{}
\pretocmd{\subsection}{\Needspace{5\baselineskip}}{}{}
\newcommand{\R}{\mathbb R}
\newcommand{\E}{\mathbb E}
\newcommand{\Prob}{\mathbb P}
\newcommand{\one}{\mathbf 1}
\newcommand{\cI}{\mathcal I}
\newcommand{\Reg}{\operatorname{SwapReg}}
\newcommand{\ExtReg}{\operatorname{ExtReg}}
\newcommand{\norm}[1]{\left\lVert#1\right\rVert}
\newcommand{\ip}[2]{\left\langle#1,#2\right\rangle}
\newcommand{\rownorm}[1]{\norm{#1}_{\infty,2}}
\newcommand{\sourcenorm}[1]{\norm{#1}_{\infty,1}}
\newcommand{\sourcenormT}[1]{\norm{#1}_{\infty,1;T}}
\newcommand{\kernorm}[1]{\norm{#1}_{\mathrm{ker},1}}
\newcommand{\shift}{\mathrm{S}}
\newcommand{\diff}{\mathrm{D}}
\newcommand{\res}{\mathrm{R}}
\newcommand{\filt}{\mathrm{F}}
\newcommand{\errfilt}{\mathrm{H}}
\newcommand{\earlyfilter}{\mathrm{K}}
\newcommand{\latefilter}{\mathrm{L}}
\newcommand{\treeP}{\Prob^{\mathrm{tr}}}
\newcommand{\quadform}{\bm K}
\newcommand{\Breg}{D_{\Phi}}
\let\epsilon\varepsilon
\DeclareMathOperator{\Var}{Var}
\DeclareMathOperator{\Cov}{Cov}
\AtBeginEnvironment{thebibliography}{\interlinepenalty=10000}
\title{\vspace{-1em}\LARGE\bfseries Constant Swap Regret in General-Sum Games\\[3pt]
via Two-Scale Higher-Order Optimism}
\author{Tung Mai\\Adobe Research}
\date{}

\begin{document}
\maketitle
\begin{abstract}
\normalsize
We give deterministic and uncoupled learning dynamics for finite multiplayer general-sum games under full-information feedback that achieve constant individual swap regret in self-play, independent of the horizon $T$. With $n$ players and at most $m$ actions each, every player’s individual swap regret is $O(\sqrt n\,m\log m\log^{5/2}(nm))$ at every finite horizon.
The dynamics use the classical Blum--Mansour framework with optimism. Each player predicts the deviation gains, uses these predictions to update a row-stochastic transition matrix, and plays its stationary distribution. Our new ingredients include a tailored row normalization map and a two-scale higher-order predictor. 
An adversarially robust variant, obtained through a generic common-prefix switching wrapper, preserves the self-play bound up to a universal constant and guarantees individual swap regret at most $7\sqrt{mT\log m}$ in the adversarial setting.
\end{abstract}
\section{Introduction}\label{sec:intro}

Online learning provides a general framework for sequential decision-making in an unknown and potentially changing environment: at each round, a learner acts, observes feedback, and updates its behavior. Its central performance measure is regret, which measures the learner's performance against a prescribed class of hindsight benchmarks. This viewpoint is especially natural in repeated games, where each player is an online learner and the feedback faced by one player is generated by the evolving behavior of the others. Correspondingly, no-regret guarantees translate into convergence toward game-theoretic equilibrium notions.

External regret compares with the best fixed action in hindsight, while swap regret allows a deviation map and compares against the best fixed deviation map in hindsight. If every player has small average swap regret, the time-averaged product distribution of play is an approximate correlated equilibrium (CE) \citep{aumann1974,blum2007}. A constant individual swap-regret bound gives an $O(1/T)$ CE rate (see Appendix~\ref{app:ce}).

In the adversarial setting, the Blum--Mansour reduction \citep{blum2007} gives $O(\sqrt{mT\log m})$ swap regret. In the self-play setting, however, the players mutually determine each other's payoff vectors. This interaction allows better swap regret bounds \citep{chen2020,anagnostides2022ce,anagnostides2022swap,tsuchiya2026}. In light of the recent constant external-regret results \citep{liu2026, abbadi2026}, we ask whether swap-regret admits the same guarantee:  

\begin{adjustwidth}{1cm}{1cm}
\emph{Can uncoupled learning dynamics achieve $O_{n,m}(1)$ individual swap regret in general-sum games, while retaining an $O(\sqrt{mT \log m})$ swap regret in the adversarial setting? 
}
\end{adjustwidth}

\paragraph{Results.}
Under full-information feedback, we give such learning dynamics. Let $n$ be the number of players and $m_i$ be the number of actions of player $i$. For $m_i\le m$, every player's individual swap regret is
$
 O\!\left(\sqrt n\,m\log m\log^{5/2}(nm)\right)
$
at every finite horizon. The dynamics are horizon-free: they do not need to know \(T\) in advance. An adversarially robust variant, obtained through a generic common-prefix switching wrapper, preserves the self-play bound up to a universal constant and gives individual swap regret at most $7\sqrt{m_i T \log m_i}$ in the adversarial setting.

\paragraph{Technical Overview.} We use the Blum--Mansour (BM) reduction for swap regret. For each source action $a$, a player maintains a distribution over the destination actions $b$. These distributions form the rows of a transition matrix, and the player uses its stationary distribution as its mixed strategy. The deviation gain assigned to changing action \(a\) to action \(b\) is \(x_{i,a}(v_{i,b}-v_{i,a})\). Therefore, controlling the best deviation in each row controls the player’s swap regret. 
We update each row optimistically by combining cumulative gains
with a forecast of the next gain into scores. A row normalization then converts these scores into a
distribution over destination actions.

Let \(\eta_i\) be player \(i\)’s learning rate. Let \(E^2\) be an upper bound on each player’s cumulative squared error in forecasting its row gains. Let \(S_i\ge0\) denote player \(i\)'s cumulative row curvature, defined using the potential associated with our row normalization. Concretely, this curvature is the local Hessian quadratic form of a row potential, evaluated in the direction of the observed row gain, and therefore measures how sensitively the row probabilities respond to changes in their scores. For our choice of row normalization, the optimistic potential analysis gives
\begin{equation}
\label{eq:regret-energy}
\eta_i\operatorname{SwapReg}_i(T)+S_i \le 2 m_i \log m_i +\eta_i^2E^2.  
\end{equation}

The main remaining task is to bound \(E\) independently of \(T\). A natural idea is to use a learner from a constant external-regret result in every BM row. However, those results rely on the feedback generated when all players follow the prescribed dynamics in a fixed game \citep{liu2026,abbadi2026}. They do not guarantee constant regret for arbitrary feedback sequences.

In BM, each row learner chooses a distribution over destination
actions, but this is not the strategy played in the game.
The player instead uses the stationary distribution of the entire
transition matrix. Therefore, changing one row can change the stationary probabilities and hence the feedback to other rows. Thus the row learners interact differently from players running the original external-regret dynamics. BM still converts the row learners' external regret into swap regret,
but their regret on these new feedback sequences requires a
separate analysis. Earlier BM analyses use the Markov-chain tree theorem \citep{anantharam1989} to
address this issue \citep{anagnostides2022ce,anagnostides2022swap}.

We build on this tree representation while extending the higher-order
optimism of \citet{liu2026,abbadi2026} to BM.
Our predictor's error is a filtered higher-order difference of
the deviation gains, so bounding it requires understanding how
these gains change as the rows are updated.
The additional difficulty is that successive differences also
act on the stationary distribution, producing interactions
between different rows of the same player.
To bound \(E\) independently of \(T\), our analysis must relate
these interactions to the row curvature supplied by the
potential bound.
At the same time, we must retain the desired dependence on
\(n\) and \(m\).
Grouping the changes in the gains only by player is too coarse
to capture the interactions between rows.
On the other hand, working immediately with all individual
transitions would incur the full action-dependent sensitivity
of the stationary distribution and lead to poor dependence
on \(m\).

The Markov-chain tree theorem represents stationary play as an average over rooted directed trees. 
The first stage keeps the change grouped by player and exploits centering across players.
Concretely, differentiating the product-tree distribution yields centered player
contributions that are independent across players, so an $L^2$ moment
bound combines their coefficients in quadrature. Thus, the learning rates
combine through
\(\bigl(\sum_i\eta_i^2\bigr)^{1/2}\), rather than
\(\sum_i\eta_i\), which preserves a final \(\sqrt n\) dependence.
After \(\ell=\Theta(\log(nm))\) such backward-difference
steps, the total weight of the terms that still require a transition-level
sensitivity bound has been attenuated by a factor of order
\(\kappa^\ell\), where \(0<\kappa<1\) is a sufficiently small
universal constant. It is then small
enough to apply the more expensive,
action-dependent curvature bound without spoiling the dependence on \(m\).

In the second stage, the proof records the individual transition indices \((i,a,b)\) encountered as further differences are taken. If an index has appeared before, its new contribution can be paired with the earlier one, and this pair is bounded by the curvature energy. Thus only a newly encountered transition index needs to be tracked further. Since there are only \(N=\sum_i m_i^2\) transition indices, there
can be at most \(N\) fresh-index continuations. The rooted-tree representation also provides the sparsity needed to avoid extra factors of \(m\): each sampled gain is associated with a single source action, and a directed tree contains at most one outgoing edge from that source action. The predictor’s two exponential moving average (EMA) cascades mirror these two stages: an \(\ell\)-step player-level stage followed by an \(N\)-step transition-level stage.

To close the argument, let \(Y^2\) denote a suitable weighted aggregate of the curvature energies \(S_i\). Since swap regret is nonnegative, from \eqref{eq:regret-energy}, we can show that
$$ Y^2 \le B+\alpha E^2, $$
where \(B\) and \(\alpha\) may depend on the game dimensions but not on \(T\). The analysis of the predictor gives a complementary bound
$$ E\le C+\beta E+\tau_\ell Y, $$
where \(C\) and \(\beta\) are constants, and \(\tau_\ell\) is exponentially small in the depth \(\ell=\Theta(\log(nm))\) of the first stage. Choosing the parameters so that \(\tau_\ell \sqrt{B}=O(1)\) and \(\beta+\tau_\ell \sqrt{\alpha}<1\) gives \(E=O(1)\). Substituting back in \eqref{eq:regret-energy} gives the constant swap-regret bound.

The design of the row normalization makes all parts of the
argument fit together.
Its associated potential has initialization cost
\(O(\log m_i)\) per row and supplies the curvature term in
\eqref{eq:regret-energy}.
The same map provides the derivative bounds needed in the
first stage and the bounds in terms of curvature needed in
the second stage, including when a transition repeats.
Crucially, these bounds use the same curvature that the
potential argument controls.
Section~\ref{sec:row-normalization} defines the map, and
Appendices~\ref{app:scalar-potential}
and~\ref{app:normalized-derivatives} establish its properties.

Finally, we obtain adversarial robustness through a common-prefix switching wrapper. An immediate switch from the base dynamics to adversarial dynamics would retain the self-play budget additively. We instead run adversarial dynamics for a common initial phase, initialize all base dynamics afresh, and permanently switch a player to freshly initialized adversarial dynamics if its base-phase regret crosses a given threshold. The initial phase absorbs the base budget into the all-horizon $O(\sqrt{m_iT\log m_i})$ envelope, and in the self-play setting the threshold is never triggered. 

\paragraph{Contributions relative to prior work.}
Our analysis combines established techniques. The Blum--Mansour reduction provides the underlying framework \citep{blum2007}. The use of Markov chain tree theorem to analyze stationary play appears in \citet{anagnostides2022ce,anagnostides2022swap}. In particular, \citet{anagnostides2022ce} use the tree representation to bound higher-order differences through the BM fixed point. Higher-order prediction for constant external regret and recursions that record player labels appear in \citet{liu2026,abbadi2026}.

Our main technical contribution is to make these techniques work together while controlling the dependence on the game dimensions. The row normalization supplies the derivative bounds used in both stages. The first stage groups the expansion by player and reduces the
coefficient passed to the transition-level analysis. The second stage
then tracks individual transitions and uses curvature to stop a branch
when a transition repeats. 
This yields \cref{lem:prediction}, in which the coefficient of the curvature energy is exponentially small in the first-stage depth. Section~\ref{sec:related-work} and Appendix~\ref{app:related} compare these steps with the earlier analyses.

\subsection{Further related work}\label{sec:related-work}
\paragraph{Regret and equilibrium learning.}
Correlated equilibrium originates with \citet{aumann1974}. The learning-to-CE connection was developed through calibrated learning and regret-based adaptive procedures by \citet{foster1997calibrated,hart2000simple}. The classical reduction of \citet{blum2007} converts external-regret guarantees into swap regret by maintaining one learner per source action and playing the stationary distribution of its transition matrix. More generally, \citet{greenwald2003general} study regret with respect to a prescribed deviation family, and \citet{gordon2008convex} connect this $\Phi$-regret viewpoint to external-regret learning over transformations and fixed-point computation.

\paragraph{External and swap regret in self-play.}
Optimistic learning exploits predictable payoff sequences \citep{rakhlin2013}. The analysis of \citet{syrgkanis2015fast} balances payoff variation against strategy movement and provides a central framework for predictable feedback in multiplayer games. For external regret, this line yields polylogarithmic bounds in finite games \citep{daskalakis2021} and logarithmic bounds in general convex games \citep{farina2022convexgames}. \citet{soleymani2025faster} introduce Cautious Optimism through adaptive, non-monotone learning-rate control for OMWU, achieving
logarithmic individual external regret in self-play. In a broader
treatment, \citet{soleymani2025cautious} develop a meta-algorithm
that paces an underlying optimistic FTRL learner and establishes
the same logarithmic dependence on $T$ under suitable
regularizer assumptions. More recent results obtain constant external regret through higher-order optimism \citep{liu2026,abbadi2026}.

For swap regret in the present full-information setting, \citet{chen2020} obtain $T^{1/4}$ horizon dependence with BM--OMWU and, in Appendix~D, with an optimistic fixed-point construction over all swap maps. \citet{anagnostides2022ce} give polylogarithmic bounds through BM--OMWU and SL--OMWU, the latter using the internal-regret reduction of \citet{stoltz2005internal}. \citet{anagnostides2022swap} obtain a logarithmic bound with BM--OFTRL. \citet{tsuchiya2026} further improves the BM--OFTRL horizon dependence to a sublogarithmic bound. Table~\ref{tab:comparison} summarizes these results. Our result is a constant swap regret bound independent of $T$. At the same time, we achieve the best polynomial dependence on both $n$ and $m$ over the previous $o_{n,m}(\log T)$ results.

\paragraph{Comparator-adaptive guarantees and alternative reductions.}
\citet{lu2025sparsity} interpolate external, internal, and swap regret through comparator sparsity. Building on this line, \citet{hait2025} improve comparator-adaptive $\Phi$-regret guarantees and obtain constant regret in a class of general-sum games satisfying \emph{nonnegative social external regret}, namely $\sum_i \ExtReg_i(T)\ge 0$ for every play sequence and horizon under the unclipped external-regret convention. This class was studied by \citet{anagnostides2022lastiterate}. Our constant swap-regret guarantee holds in arbitrary finite general-sum games. Beyond the classical per-action BM reduction, \citet{dagan2024external,peng2024fast} achieve any fixed target level of average swap regret with round complexity polylogarithmic in the number of actions. These action--accuracy tradeoffs are distinct from a constant cumulative swap-regret guarantee in self-play. In convex domains, \citet{anagnostides2026response} use response-based approachability to minimize linear and profile swap regret and extend their results to polynomial-dimensional deviation families. These broader comparator classes lie outside the finite normal-form swap-regret setting studied here. Finally, \citet{tsuchiya2026scale} obtain scale-free and scale-invariant swap regret $O(U_{\max}n^{3/2}m^{5/2}\log T)$ in general-sum games.

\paragraph{Adversarial robustness.}
Adversarial robustness in this line is obtained through closely related certificate-triggered constructions. Earlier work describes these as a wrapper or an adaptive choice of learning rate \citep{chen2020,daskalakis2021,anagnostides2022ce}, while the closest recent constructions use a switching rule that switches a player permanently to adversarial dynamics after a self-play certificate fails \citep{anagnostides2022swap,tsuchiya2026,abbadi2026}. Such an immediate switch retains the self-play budget additively. We add a common phase before a fresh simultaneous start of the base dynamics, allowing that budget to be absorbed into the all-horizon $O(\sqrt{mT \log m})$ adversarial guarantee. Our approach is black-box in the underlying self-play dynamics. Therefore, it can strengthen existing results with anytime self-play regret bounds by removing their additive self-play budget whenever a common-prefix absorption condition holds.

\begin{table}[!htbp]
\centering
\scriptsize
\setlength{\tabcolsep}{3pt}
\begin{tabularx}{\textwidth}{@{}>{\raggedright\arraybackslash}p{0.21\textwidth}>{\raggedright\arraybackslash}p{0.22\textwidth}>{\raggedright\arraybackslash}X>{\raggedright\arraybackslash}X@{}}
\toprule
Reference & Algorithm & Self-play setting & Adversarial setting\\
\midrule
\citet{blum2007} & BM--MWU & $O(\sqrt{mT\log m})$ & $O(\sqrt{mT\log m})$\\[5pt]
\citet{chen2020} & BM--OMWU & $O(\sqrt n\,(m\log m)^{3/4}T^{1/4})$ & $\widetilde{O}(\sqrt{mT}+\sqrt n\,m^{3/4}T^{1/4})$\\[5pt]
\citet[Appendix~D]{chen2020} & Optimistic swap-map Hedge (all $m^m$ maps) & $O(\sqrt n\,m^{5/4}(\log m)^{3/4}T^{1/4})$ & $O(\sqrt{mT\log m})$\\[5pt]
\citet{anagnostides2022ce} & SL--OMWU / adaptive variant (via internal regret) & $O(nm\log m\log^4T)$ & $O(nm\log m\log^4T+m\sqrt{T\log m})$\\[5pt]
\citet{anagnostides2022ce} & BM--OMWU & $O(nm^4\log m\log^4T)$ & ---\\[5pt]
\citet{anagnostides2022swap} & BM--OFTRL (log-barrier) & $O(nm^{5/2}\log T)$ & $O(nm^{5/2}\log T+\sqrt{mT\log m})$\\[5pt]
\citet{tsuchiya2026} & BM--OFTRL (hybrid) & $O(nm^2\sqrt{\log m\log T})$ & $O(nm^2\sqrt{\log m\log T}+\sqrt{mT\log m})$\\[5pt]
\textbf{This work} & BM--OFTRL with two-scale higher-order optimism & $O\!\left(\sqrt n\,m\log m\log^{5/2}(nm)\right)$ & $O(\sqrt{mT\log m})$\\
\bottomrule
\end{tabularx}
\caption{Comparisons of selected algorithms for individual swap regret in full-information $n$-player games with at most $m$ actions. BM denotes Blum--Mansour, and SL denotes Stoltz--Lugosi. $\widetilde{O}$ suppresses logarithmic factors in $m,n,T$.}
\label{tab:comparison}
\end{table}
\FloatBarrier

\section{Model and main results}\label{sec:model-results}
We may assume each player $i\in[n]$ has $m_i\ge2$ actions and a fixed payoff function $u_i:\prod_j[m_j]\to[0,1]$. Let $m=\max_i m_i$. At round $t$, the players simultaneously choose mixed strategies $\bm x_i^{(t)}\in\Delta_{m_i}$, and player $i$ observes the expected payoff 
$
 v_{i,a}^{(t)}=u_i(a,\bm x_{-i}^{(t)}), \forall a\in[m_i].
$
This is full-information feedback since player $i$ observes the exact expected payoff of every action. The learning dynamics are uncoupled and horizon-free: player $i$ uses only public action counts, its own state, and past payoff vectors, and does not require the final horizon $T$. Each mixed strategy is a deterministic output of the dynamics. The self-play setting is when every player follows the same dynamics. In the adversarial setting, the payoff vector may instead be any arbitrary vector in $[0,1]^{m_i}$ each round.

Let $\cI_i=\{(i,a,b):a,b\in[m_i]\}$, $\cI=\bigcup_i\cI_i$, and $N=|\cI|=\sum_im_i^2$. We call each $e\in\cI$ a transition index. Define pairwise deviation gains and player $i$'s individual swap regret by
\begin{align*}
 r_{i,a,b}^{(t)}&=x_{i,a}^{(t)}(v_{i,b}^{(t)}-v_{i,a}^{(t)}),\\
 \Reg_i(T)&=\max_{\varphi:[m_i]\to[m_i]}\sum_{t=1}^T\sum_a r_{i,a,\varphi(a)}^{(t)}
 =\sum_a\max_b\sum_{t=1}^T r_{i,a,b}^{(t)}\ge0.
\end{align*}
The equality follows because the deviation can be chosen independently for each action. Nonnegativity follows by including the identity map.

The proof uses two mixed row norms. For each player $i$ and each $\bm w_i = (w_{i,a,b})_{a,b \in [m_i]}$, let
\begin{equation*}
 \rownorm{\bm w_i}^{\,2}=\sum_a\max_b|w_{i,a,b}|^2,
 \qquad
 \sourcenorm{\bm w_i}=\sum_a\max_b|w_{i,a,b}|,
 \qquad
 \sourcenormT{\bm w_i}^2=\sum_{t=1}^T \sourcenorm{\bm w_i^{(t)}}^2.
\end{equation*}
We have $\rownorm{\bm w_i}\le\sourcenorm{\bm w_i}$ and $\sourcenorm{\bm r_i^{(t)}}\le1$. 

\paragraph{Public parameters.} 
\begin{equation}\label{eq:public-scales}
 A_i=m_i\log m_i,\quad A=\sum_iA_i,\quad
 \ell=\ell_0+\lceil\log(2+n+N)\rceil,\quad k=\ell+1,\quad g=c/k^{5/2}.
\end{equation}
Here $c>0$ is a sufficiently small universal constant and $\ell_0$ is a sufficiently large universal integer. The default rates are
\begin{equation}\label{eq:default-rates}
 \eta_i=g\sqrt{A_i/A},\qquad \sum_i\eta_i^2=g^2.
\end{equation}
We call the learning dynamics of \cref{alg:base} the \emph{base dynamics}.

\begin{theorem}[Constant swap regret]\label{thm:base}
There exist universal choices of $c>0$ and $\ell_0\in\mathbb N$ such that, for every finite game, if all players simultaneously use \cref{alg:base} with the public parameters above, then every player $i$ satisfies, for every $T\ge1$,
\begin{equation*}
 \Reg_i(T)\le \frac4c k^{5/2}\sqrt{A_iA}
 =O\!\left(\sqrt n\,m\log m\log^{5/2}(nm)\right).
\end{equation*}
These learning dynamics are deterministic, uncoupled, and horizon-free.
\end{theorem}

\begin{remark}[Unequal action counts]\label{rem:unequal-actions}
For the base dynamics, one may instead set $U=\sum_jA_j^2$ and $\eta_i=gA_i/\sqrt U$. The same proof gives $\max_i\Reg_i(T)\le \frac4c k^{5/2}\sqrt U$ for every prefix. This equals the default worst-player bound for equal action counts and can improve it by a factor $\Theta(n^{1/4})$ for highly unequal action counts. Appendix~\ref{app:heterogeneous-comparison} gives the details.
\end{remark}

\begin{theorem}[Adversarial robustness]\label{thm:robust}
The dynamics obtained by instantiating \cref{alg:robust} with \cref{alg:base} and the default parameters satisfy: in the adversarial setting,
\[
 \Reg_i(T)\le7\sqrt{A_iT}
 =O\!\left(\sqrt{mT\log m}\right)
\]
and in the self-play setting,
\[
 \Reg_i(T)\le \frac{16}{c}k^{5/2}\sqrt{A_iA}
 =O\!\left(\sqrt n\,m\log m\log^{5/2}(nm)\right)
\]
\end{theorem}

\section{Optimistic transition-matrix dynamics}\label{sec:algorithm}
This section presents the main dynamics. In each round, player \(i\) combines a prediction of its next pairwise deviation-gain matrix with its cumulative gains into a score matrix. A smooth row-normalization function converts each score row into a probability distribution. These probability rows form a transition matrix whose stationary distribution is played. \Cref{sec:row-normalization,sec:predictor} define the normalization function and predictor, respectively, before \cref{alg:base} combines them.

Figure~\ref{fig:algorithm-flow} summarizes the feedback loop. In particular, the deviation gains depend on the player's own stationary strategy as well as its observed payoff vector.
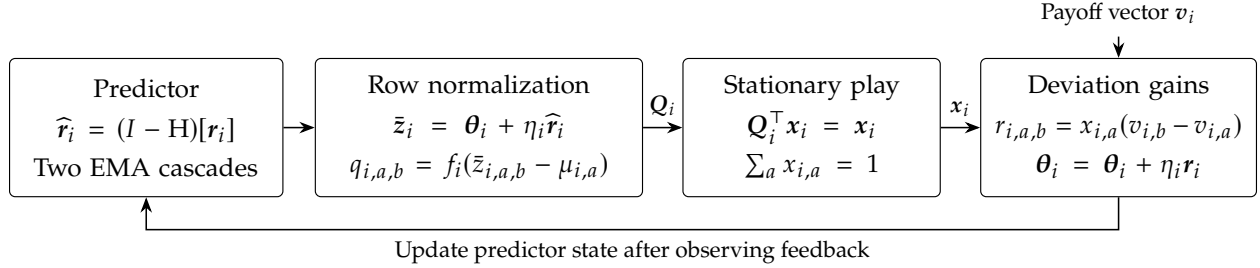
\begin{figure}[!htbp]
\centering
\resizebox{\textwidth}{!}{\begin{tikzpicture}[x=1cm,y=1cm,>=Stealth,
  every node/.style={font=\footnotesize},
  block/.style={draw,rounded corners=2pt,align=center,minimum height=1.72cm,
                text width=3.05cm,inner sep=5pt},
  flow/.style={->,line width=0.55pt}]
\node[block] (pred) at (0,0)
  {Predictor\\[3pt]
   $\widehat{\bm r}_i=(I-\errfilt)[\bm r_i]$\\[3pt]
   Two EMA cascades};
\node[block,text width=3.72cm] (row) at (4.14,0)
  {Row normalization\\[3pt]
   $\bar{\bm z}_i=\bm\theta_i+\eta_i\widehat{\bm r}_i$\\[3pt]
   $q_{i,a,b}=f_i(\bar z_{i,a,b}-\mu_{i,a})$};
\node[block,text width=2.85cm] (stat) at (8.28,0)
  {Stationary play\\[3pt]
   $\bm Q_i^\top\bm x_i=\bm x_i$\\[3pt]
   $\sum_a x_{i,a}=1$};
\node[block,text width=3.08cm] (gain) at (12.10,0)
  {Deviation gains\\[3pt]
   $r_{i,a,b}= x_{i,a}(v_{i,b}-v_{i,a})$\\[3pt]
   $\bm\theta_i=\bm\theta_i+\eta_i{\bm r}_i$};
\draw[flow] (pred)--(row);
\draw[flow] (row)--node[above,font=\scriptsize]{$\bm Q_i$}(stat);
\draw[flow] (stat)--node[above,font=\scriptsize]{$\bm x_i$}(gain);
\node[font=\scriptsize] (pay) at (12.10,1.41) {Payoff vector $\bm v_i$};
\draw[flow] (pay)--(gain.north);
\draw[flow] (gain.south)--(12.10,-1.27)--(0,-1.27)--(pred.south);
\node[font=\scriptsize] at (6.05,-1.55)
  {Update predictor state after observing feedback};
\end{tikzpicture}}
\caption{\textbf{Optimistic transition-matrix dynamics}. Each row is normalized through $f_i$ with its own shift $\mu_{i,a}$, and the player uses the resulting stationary strategy. $\bm\theta_i$ denotes the cumulative score.}
\label{fig:algorithm-flow}
\end{figure}
\FloatBarrier

\subsection{Smooth row normalization}\label{sec:row-normalization}
First we explain the normalization process for a score row $\bm z \in \R^{m_i}$. 
The normalization map must satisfied certain properties to support both a potential argument and a higher-order prediction argument later in the proof. These properties are established in
Lemmas~\ref{lem:scalar-bounds}, \ref{lem:potential-full},
and~\ref{lem:normalized-derivatives-full}.
For $s\in\R$, define
\begin{equation}\label{eq:scalar-response}
 u(s)=\frac{\sqrt{s^2+4}-s}{2},
 \qquad t_+(s)=\frac{\sqrt{s^2+4}+s}{2}=u(s)^{-1},
 \qquad
 f_i(s)=\beta_i\frac{(1+t_+(s)/d_i)^{d_i}}{1+u(s)+(\delta/4)u(s)^2}.
\end{equation}
$f_i$ is smooth, positive, strictly increasing, and ranges from zero to infinity. For $\bm z\in\R^{m_i}$, define
\begin{equation}\label{eq:row-response}
 q_{i,b}(\bm z)=f_i(z_b-\mu),
 \qquad \sum_bf_i(z_b-\mu)=1.
\end{equation}
For each score row $\bm z$, $\mu=\mu_i(\bm z)$ is its own unique normalization shift, so that $\sum_b f_i(z_b-\mu)=1$.  

\subsection{Two EMA cascades for the predictor}\label{sec:predictor}
We use higher-order prediction with EMA filters, following \citet{liu2026}. The two filter depths match the two parts of our proof: $\ell$ for the stage that groups terms by player and $N$ for the stage that records individual transitions. 

Let $\shift$ be the backward shift and $\diff=I-\shift$, with zero prehistory. For $r\ge1$, let $\rho_r=\frac r{r+1}$. We define the following operators:
\begin{equation*}
 \qquad \res_r=(I-\rho_r\shift)^{-1},
 \qquad \filt_r=\res_r\diff,
 \qquad \errfilt=\diff^3\filt_\ell^\ell\filt_N^N.
\end{equation*}
The predictor is $\widehat{\bm r}=(I-\errfilt)[\bm r]$. These are exponential moving average (EMA) cascades. Every factor of $\errfilt$ has current-input coefficient one, so the current input cancels in $I-\errfilt$ and the predictor depends only on past inputs. The state variables in \cref{alg:base} give an implementation of this predictor.

\begin{algorithm}[!htbp]
\caption{Base learning dynamics for player $i$}\label{alg:base}
\small
\textbf{Input.} Public parameters from \cref{eq:public-scales,eq:default-rates} and Appendix~\ref{app:parameters}, with row normalization \eqref{eq:row-response}.\\
\textbf{State.} Initialize the $m_i\times m_i$ arrays $\bm\theta,\bm a_1,\ldots,\bm a_{N+\ell},\bm c_1,\bm c_2,\bm c_3$ to zero. Let $\nu_j=(N+1)^{-1}$ for $j\le N$ and $\nu_j=(\ell+1)^{-1}$ for $j>N$.
\begin{algorithmic}[1]
\For{$t=1,2,\ldots$}
 \State $\widehat{\bm r}\gets\sum_{j=1}^{N+\ell}\bm a_j+3\bm c_1-3\bm c_2+\bm c_3$
 \State $\overline{\bm z}\gets\bm\theta+\eta_i\widehat{\bm r}$
 \Comment{raw row score}
 \State Normalize each row of $\overline{\bm z}$ by \eqref{eq:row-response} or Appendix~\ref{app:solver}. Form $\bm Q$.
 \Comment{row normalization}
 \State Solve $\bm Q^\top\bm x=\bm x$, $\sum_ax_a=1$, and play $\bm x$.
 \State Observe $\bm v$ and set $r_{a,b}\gets x_a(v_b-v_a)$ for all $a,b$.
 \State $\bm\theta\gets\bm\theta+\eta_i\bm r$, then $\bm u\gets\bm r$
 \For{$j=1,\ldots,N+\ell$}
  \State $\bm u\gets\bm u-\bm a_j$, then $\bm a_j\gets\bm a_j+\nu_j\bm u$
 \EndFor
 \State $(\bm c_3,\bm c_2,\bm c_1)\gets(\bm c_2,\bm c_1,\bm u)$ \Comment{simultaneously}
\EndFor
\end{algorithmic}
\end{algorithm}

\begin{remark}\label{rem:implementation}
In the exact-real computation model of Appendix~\ref{app:solver},
with public parameters supplied as exact constants, \cref{alg:base}
uses polynomially many operations per round in $n$, $m$, and
$\log t$ for each player.
\end{remark}

Write $\bm a_j^{(t)}$ for the state entering round $t$ and
$\bm u_j^{(t)}$ for the residual produced after cascade stage $j$
in that round, with $\bm u_0^{(t)}=\bm r^{(t)}$ and
$\bm c^{(t)}=\bm u_{N+\ell}^{(t)}$. The update gives
$[I-(1-\nu_j)\shift]\bm u_j=\diff\bm u_{j-1}$, where
$\nu_j=(N+1)^{-1}$ in the first cascade and
$\nu_j=(\ell+1)^{-1}$ in the second. Hence
\begin{equation*}
 \bm c=\filt_\ell^\ell\filt_N^N\bm r,
 \qquad
 \sum_j\bm a_j^{(t)}=\bm r^{(t)}-\bm c^{(t)},
 \qquad
 \bm r-\widehat{\bm r}=\diff^3\bm c=\errfilt[\bm r].
\end{equation*}
Appendix~\ref{app:filters} proves the depth-uniform filter bounds.

\subsection{Cumulative, raw, and effective scores}\label{sec:effective-scores}
We define the cumulative score and the raw score passed to row normalization, respectively:
\[
\boldsymbol{\theta}_i^{(t)}
    := \eta_i \sum_{s \le t} \boldsymbol{r}_i^{(s)},
\qquad
\bar{\boldsymbol{z}}_i^{(t)}
    := \boldsymbol{\theta}_i^{(t-1)}
       + \eta_i \widehat{\boldsymbol{r}}_i^{(t)}
\]
Appendix~\ref{app:solver} shows that the finite row normalization solver returns an effective score
\begin{equation}\label{eq:solver-interface}
 \bm z_i^{(t)}=\overline{\bm z}_i^{(t)}+\eta_i\bm\varepsilon_i^{(t)},
 \qquad
 \sourcenorm{\bm\varepsilon_i^{(t)}}\le(t+1)^{-2},
 \qquad
 \sum_{s=1}^T\sourcenorm{\bm\varepsilon_i^{(s)}}^2\le1.
\end{equation}
Exact normalization corresponds to $\bm\varepsilon_i^{(t)}=0$. 
Using
$\widehat{\boldsymbol r}_i=(I-\errfilt)[\boldsymbol r_i]$ and
$\boldsymbol\theta_i^{(t)}
 =\boldsymbol\theta_i^{(t-1)}
  +\eta_i\boldsymbol r_i^{(t)}$,
the effective score can be rewritten as
\[
\boldsymbol z_i^{(t)}
=
\boldsymbol\theta_i^{(t)}
-\eta_i\left(
\errfilt[\boldsymbol r_i]^{(t)}
-\boldsymbol\varepsilon_i^{(t)}
\right).
\]
We therefore define the effective prediction error
\[
\boldsymbol\xi_i
:=\errfilt[\boldsymbol r_i]-\boldsymbol\varepsilon_i.
\]
Consequently,
\begin{equation}\label{eq:effective-score-difference}
\boldsymbol z_i=\boldsymbol\theta_i-\eta_i\boldsymbol\xi_i,
\qquad
\diff\boldsymbol z_i
=\eta_i\bigl(\boldsymbol r_i-\diff\boldsymbol\xi_i\bigr).
\end{equation}

The filter bounds of Appendix~\ref{app:filters} and \eqref{eq:solver-interface} imply, for arbitrary payoff-vector sequences,
\begin{equation}\label{eq:pointwise-bounds}
 \sourcenorm{\widehat{\bm r}_i^{(t)}}\le97,
 \qquad
 \sourcenorm{\bm\xi_i^{(t)}}\le97.
\end{equation}

\FloatBarrier

\section{Proof of the constant swap-regret theorem}
\label{sec:proof-sketch}

Fix a finite horizon $T$. We prove \cref{thm:base} for any fixed positive public rates satisfying
\begin{equation}\label{eq:rate-budget}
 \sum_i\eta_i^2\le g^2.
\end{equation}
The default rates in \eqref{eq:default-rates} satisfy this condition.
Recall that $\bm\xi_i=\errfilt[\bm r_i]-\bm\varepsilon_i$ is
player $i$'s effective prediction error, and define its cumulative
size by
\[
 E^2=\max_i\sum_{t=1}^T
       \sourcenorm{\bm\xi_i^{(t)}}^2.
\]
For one normalized row $\bm q=\bm q_i(\bm z)$, define
$h_i(f_i(s))=f_i'(s)$ and
\begin{equation}\label{eq:row-geometry}
 h_b=h_i(q_b),
 \qquad
 \sigma_b=\frac{h_b}{\sum_ch_c},
 \qquad
 \widetilde w_b=w_b-\ip{\bm\sigma}{\bm w},
 \qquad
 \quadform_{\bm z}(\bm w)=\sum_bh_b\widetilde w_b^2.
\end{equation}
Thus the row curvature is centered under normalized curvature
weights. At the effective row scores, set
\[
 b_{i,a}^{(t)}
 =\frac{\eta_i^2}{8}
   \quadform_{\bm z_{i,a}^{(t)}}(\bm r_{i,a}^{(t)}),
 \qquad
 S_i=\sum_{t=1}^T\sum_ab_{i,a}^{(t)}.
\]
Let $\beta_i$ be the public row-normalization parameter from
Appendix~\ref{app:parameters}, let $C_\omega$ be the universal proof
constant fixed in Appendix~\ref{app:constant-selection}, and define
\[
 \omega_i=C_\omega\frac{m_i-1}{\beta_i},
 \qquad
 Y^2=\sum_i\omega_iS_i,
 \qquad
 V=\sum_i\omega_iA_i,
 \qquad
 \Lambda^2=\sum_i\omega_i\eta_i^2,
 \qquad
 \eta_{\max}=\max_i\eta_i.
\]

The proof rests on two bounds relating $E$ and $Y$. First, the
row-potential argument gives
\[
 \eta_i\Reg_i(T)+S_i\le2A_i+\eta_i^2E^2,
 \qquad
 Y^2\le2V+\Lambda^2E^2.
\]
Second, the self-play prediction argument gives
\[
 E\le C_0+C_1gE+\tau_\ell Y,
 \qquad
 \kappa:=C_Egk^{5/2},
 \qquad
 \tau_\ell:=C_Pk^2\kappa^\ell,
\]
for universal constants $C_E,C_0,C_1,C_P$. The factor $\tau_\ell$
is crucial: the dimension-dependent quantity $Y$ enters only after
its coefficient has been reduced exponentially in the first-stage
depth $\ell$. The public parameter choices allow the two bounds to
be absorbed, yielding a constant bound on $E$ and hence the claimed
swap-regret bound.

Section~\ref{sec:two-estimates} states the two bounds formally and
carries out the absorption. Section~\ref{sec:energy-idea} proves the
row-potential bound.
Sections~\ref{sec:two-stage-idea}--\ref{sec:filter-assembly} prove the
prediction bound: Section~\ref{sec:two-stage-idea} introduces the
stationary-tree representation, and
Section~\ref{sec:prediction-setup} gives the shared differentiation
and filtering setup. Section~\ref{sec:early-sketch} performs the
first $\ell$ player-centered differences: it bounds some terms
immediately and passes the others to the late stage with total
absolute weight at most $\kappa^\ell$, after accounting for the
early filter. Section~\ref{sec:late-sketch} records transition
indices, charging repeated indices to $Y$ while only fresh indices
continue. Section~\ref{sec:filter-assembly} combines the two stage
bounds to prove the prediction inequality.

\subsection{The energy and prediction bounds imply the theorem}
\label{sec:two-estimates}

We first record the two bounds precisely and close the feedback loop
between $E$ and $Y$. The first bound is valid for arbitrary
payoff-vector sequences. It controls each player's swap regret and
bounds the weighted curvature energy in terms of the prediction error.

\begin{lemma}\label{lem:energy}
For sufficiently small universal $c$, all payoff-vector sequences
satisfying $\bm v_i^{(t)}\in[0,1]^{m_i}$ for every player $i$ and
round $t$ obey
\begin{equation}\label{eq:energy}
 \eta_i\Reg_i(T)+S_i\le2A_i+\eta_i^2E^2,
 \qquad
 Y^2\le2V+\Lambda^2E^2,
 \qquad
 \Lambda^2\le g^2V.
\end{equation}
For the default rates, $\Lambda^2=g^2V/A$.
\end{lemma}

The second bound is specific to self-play and runs in the
complementary direction:

\begin{lemma}\label{lem:prediction}
There are universal constants $C_E,C_0,C_1,C_P$ such that, under
\eqref{eq:rate-budget} and the public parameter bounds of
Appendix~\ref{app:parameters}, if $\kappa=C_Egk^{5/2}\le1/16$ and
$\tau_\ell=C_Pk^2\kappa^\ell$, then the following holds in the
self-play setting:
\begin{equation}\label{eq:prediction}
 E\le C_0+C_1gE+\tau_\ell Y.
\end{equation}
\end{lemma}
The factor $\tau_\ell$ is essential because $V$, and hence the
a priori size of $Y$, depends on the game dimensions.
Appendix~\ref{app:constant-selection} makes $\tau_\ell\sqrt V$
uniformly bounded. In particular, Appendix~\ref{app:constant-selection} selects the universal
constants so that
\begin{equation}\label{eq:parameter-absorption}
 \kappa\le\frac1{16},
 \qquad
 \tau_\ell\sqrt{2V}\le1,
 \qquad
 (C_1+1)g\le\frac12.
\end{equation}
Combining \eqref{eq:energy} and \eqref{eq:prediction} now closes
the feedback loop. By \cref{lem:energy},
\[
 Y\le\sqrt{2V}+\Lambda E
   \le\sqrt{2V}+g\sqrt V E.
\]
Hence $\tau_\ell Y\le1+gE$, and \cref{lem:prediction} gives
\[
 E\le C_0+1+(C_1+1)gE,
 \qquad
 E\le E_0:=2(C_0+1).
\]
The choice of $c$ also ensures $\eta_i^2E_0^2\le1\le A_i$.
Returning to \eqref{eq:energy} and discarding $S_i\ge0$ yields
\begin{equation}\label{eq:generic-bound}
 \Reg_i(T)\le\frac{3A_i}{\eta_i}.
\end{equation}
For the default rates this is
$3k^{5/2}\sqrt{A_iA}/c\le(4/c)k^{5/2}\sqrt{A_iA}$, which gives
\cref{thm:base}. It remains to prove \cref{lem:energy} in
Section~\ref{sec:energy-idea} and \cref{lem:prediction} in
Sections~\ref{sec:two-stage-idea}--\ref{sec:filter-assembly}.

\subsection{Stationarity yields the energy bound}
\label{sec:energy-idea}

This subsection proves \cref{lem:energy}. Define
\begin{equation}\label{eq:potential}
 P_i(s)=\int_{-\infty}^s f_i(v)\,dv,
 \qquad
 c_{\Phi,i}=\min_s\{P_i(s)-s\},
 \qquad
 \Phi_i(\bm z)
 =\mu_i(\bm z)
  +\sum_bP_i(z_b-\mu_i(\bm z))-c_{\Phi,i}.
\end{equation}
The construction in Appendix~\ref{app:scalar-potential} gives
\begin{equation}\label{eq:potential-properties}
 \nabla\Phi_i(\bm z)=\bm q_i(\bm z),
 \qquad
 \Phi_i(\bm z)\ge\max_bz_b,
 \qquad
 0\le\Phi_i(0)\le2\log m_i,
 \qquad
 \bm w^\top\nabla^2\Phi_i(\bm z)\bm w
 =\quadform_{\bm z}(\bm w),
\end{equation}
together with the local curvature comparison used below. Write
\[
 \Breg(\bm y,\bm z)
 =\Phi_i(\bm y)-\Phi_i(\bm z)
  -\ip{\bm q_i(\bm z)}{\bm y-\bm z}.
\]

Fix a row and a round, abbreviate $\bm z=\bm z_{i,a}^{(t)}$,
and set
\[
 \bm d^+
 =\bm\theta_{i,a}^{(t)}-\bm z
 =\eta_i\bm\xi_{i,a}^{(t)},
 \qquad
 \bm d^-
 =\bm\theta_{i,a}^{(t-1)}-\bm z
 =\eta_i(\bm\xi_{i,a}^{(t)}-\bm r_{i,a}^{(t)}).
\]
By \eqref{eq:pointwise-bounds}, both displacements have sup norm
at most $98\eta_i$. The universal choice of $c$ gives
$98\eta_i\le98c/k^{5/2}<1/8$, so the local comparison applies on
both segments. Taylor's integral formula gives
\[
 \Breg(\bm z+\bm d^+,\bm z)
 +\Breg(\bm z+\bm d^-,\bm z)
 \ge\frac14\bigl(
   \quadform_{\bm z}(\bm d^+)
   +\quadform_{\bm z}(\bm d^-)
 \bigr)
 \ge b_{i,a}^{(t)}.
\]
The second inequality uses $Q(x)+Q(y)\ge Q(x-y)/2$ for every
positive-semidefinite quadratic form $Q$, together with
$\bm d^+-\bm d^-=\eta_i\bm r_{i,a}^{(t)}$. Also, by
\eqref{eq:hessian-full},
\[
 2\Breg(\bm z+\bm d^+,\bm z)
 \le\norm{\bm d^+}_\infty^2.
\]
Subtracting the two Taylor expansions yields
\[
 \Phi_i(\bm\theta_{i,a}^{(t)})
 -\Phi_i(\bm\theta_{i,a}^{(t-1)})
 \le\eta_i\ip{\bm q_{i,a}^{(t)}}{\bm r_{i,a}^{(t)}}
   +\eta_i^2\norm{\bm\xi_{i,a}^{(t)}}_\infty^2
   -b_{i,a}^{(t)}.
\]
The linear terms vanish after summing the source rows, because
stationarity gives
\[
 \sum_{a,b}q_{i,a,b}^{(t)}r_{i,a,b}^{(t)}
 =(\bm x_i^{(t)})^\top
   (\bm Q_i^{(t)}-I)\bm v_i^{(t)}
 =0.
\]
Telescoping over time and using \eqref{eq:potential-properties}
gives
\[
 \sum_a\Phi_i(\bm\theta_{i,a}^{(T)})
 \ge\eta_i\sum_a\max_b\sum_{t=1}^Tr_{i,a,b}^{(t)}
 =\eta_i\Reg_i(T),
 \qquad
 \sum_a\Phi_i(0)\le2A_i.
\]
Moreover,
\[
 \sum_{t=1}^T\sum_a
   \norm{\bm\xi_{i,a}^{(t)}}_\infty^2
 \le\sum_{t=1}^T
   \left(\sum_a\norm{\bm\xi_{i,a}^{(t)}}_\infty\right)^2
 =\sum_{t=1}^T\sourcenorm{\bm\xi_i^{(t)}}^2
 \le E^2.
\]
Therefore
\[
 \eta_i\Reg_i(T)+S_i\le2A_i+\eta_i^2E^2.
\]
Multiplying by $\omega_i$, summing over players, and using
$\Reg_i(T)\ge0$ gives
\[
 Y^2\le2V+\Lambda^2E^2.
\]
Finally,
$\Lambda^2\le(\max_i\omega_i)\sum_i\eta_i^2\le g^2V$,
while the default rates give $\Lambda^2=g^2V/A$.
This proves \cref{lem:energy}.

\subsection{Probabilistic tree representation}
\label{sec:two-stage-idea}
We use the classical tree representation \citep{anantharam1989}, as in the analyses of \citet{anagnostides2022ce,anagnostides2022swap}. Each deviation gain is represented as an expectation under a product distribution on rooted directed trees. Differentiating that distribution produces a centered tree score. This representation supplies both a bound that groups contributions by player and a curvature bound for individual transitions. We use these bounds in the two stages of the proof of \cref{lem:prediction}.

At trial scores $\bm y$, let
$\treeP_{i,\bm y}$ be the normalized distribution on rooted directed trees
whose weights are products of transition entries, and let
$\treeP_{\bm y}=\bigotimes_i\treeP_{i,\bm y}$.
The tree theorem identifies each player's root distribution with
its stationary strategy. For a product-tree sample $\bm T$, define
\[
 G_{i,a,b}(\bm T)
 =\one_{\{a_i(\bm T)=a\}}
  \bigl[u_i(b,\bm a_{-i}(\bm T))
        -u_i(a,\bm a_{-i}(\bm T))\bigr].
\]
The root selects only one source row, so
$\sourcenorm{\bm G_i(\bm T)}\le1$. Let
\[
 \Gamma_i(\bm y):=\E_{\treeP_{\bm y}}\bm G_i,
 \qquad
 \bm r_i^{(t)}=\Gamma_i(\bm z^{(t)}).
\]
For \(t\ge2\), consider interpolating between the two effective scores by
\[
 \bm y_\theta^{(t)}
 =(1-\theta)\bm z^{(t-1)}+\theta\bm z^{(t)},
 \qquad 0\le\theta\le1.
\]
Since \(\bm G_i\) has no explicit score dependence, the fundamental
theorem of calculus and differentiation of the tree distribution give
\begin{align}
 \diff\bm r_i^{(t)}
 &=
 \Gamma_i(\bm z^{(t)})-\Gamma_i(\bm z^{(t-1)}) \notag\\
 &=
 \int_0^1
 \nabla_{\diff\bm z^{(t)}}
 \Gamma_i(\bm y_\theta^{(t)})\,d\theta \notag\\
 &=
 \int_0^1
 \E_{\treeP_{\bm y_\theta^{(t)}}}\!\left[
   \bm G_i(\bm T)
   \nabla_{\diff\bm z^{(t)}}
   \ln\treeP_{\bm y_\theta^{(t)}}(\bm T)
 \right]d\theta .
 \label{eq:first-gain-difference}
\end{align}
Thus a backward difference of the gain is represented exactly as an
integrated derivative of the tree distribution along the change in
effective scores.

Write $I_e$ for edge inclusion, with $I_e=0$ for diagonal indices.
Differentiating the normalized tree distribution in a deterministic
direction $\bm w$ gives
\begin{equation}\label{eq:tree-score}
 \nabla_{\bm w}\ln\treeP_{\bm y}
 =\sum_e(I_e-\E I_e)\frac{h_e}{q_e}\widetilde w_e.
\end{equation}
The uncentered score of player $i$ has magnitude at most
$2\sourcenorm{\bm w_i}$. Independence across players therefore
implies
\begin{equation}\label{eq:early-tree-score}
 \norm{\nabla_{\bm w}\ln\treeP_{\bm y}}_{L^2}
 \le2\left(\sum_i\sourcenorm{\bm w_i}^2\right)^{1/2}.
\end{equation}
Appendix~\ref{app:trees-moments} proves these identities and
bounds. The early stage uses this player-centered bound,
which does not involve $\omega_i$. The late stage retains the
transition indices in \eqref{eq:tree-score} and uses the
curvature-weighted bound.

\subsection{From gain differences to products of centered factors}
\label{sec:prediction-setup}

By the solver guarantee \eqref{eq:solver-interface},
\[
 E\le1+\max_i\sourcenormT{\errfilt[\bm r_i]}.
\]
A pointwise gain bound alone would give a bound growing as
$\sqrt T$. Instead, we use the differences already present in
$\errfilt$ and expand only the terms not yet controlled.

\paragraph{The filter supplies a backward difference.}
Since $\filt_r=\res_r\diff$ and all factors commute,
\[
 \errfilt
 =\res_\ell^\ell\res_N^N\diff^{N+\ell+3}.
\]
For $0\le p\le\ell$, define
\[
 \earlyfilter_p
 :=\res_\ell^\ell\res_N^N\diff^{N+\ell+3-p}.
\]
Then $\earlyfilter_p=\earlyfilter_{p+1}\diff$ for $p<\ell$,
and the first step is
\[
 \errfilt[\bm r_i]
 =\earlyfilter_0[\bm r_i]
 =\earlyfilter_1[\diff\bm r_i].
\]
We now follow the first two differences of the gain to see what
can be bounded and what must be expanded further.

\paragraph{The first difference separates error and gain directions.}
Fix $t\ge2$ and interpolate
$\bm y_\theta^{(t)}
 =(1-\theta)\bm z^{(t-1)}+\theta\bm z^{(t)}$.
In the following calculation, suppress the superscript $(t)$
and write $\E_\theta$ for expectation at $\bm y_\theta$.
For a product direction, use the shorthand
$\eta\bm w=(\eta_j\bm w_j)_{j=1}^n$.
Applying \eqref{eq:first-gain-difference} and using
\(\diff\bm z=\eta\bm r-\eta\diff\bm\xi\), we obtain
\[
 \begin{aligned}
  \diff\bm r_i
   =\underbrace{-\int_0^1\E_\theta\!\left[
       \bm G_i\nabla_{\eta\diff\bm\xi}
       \ln\treeP_{\bm y_\theta}
     \right]d\theta}_{\text{error direction: stop and bound using }E}
  \quad + \quad
     \underbrace{\int_0^1\E_\theta\!\left[
       \bm G_i\nabla_{\eta\bm r}
       \ln\treeP_{\bm y_\theta}
     \right]d\theta}_{\text{gain direction: continue taking derivative}}.
 \end{aligned}
\]
The error contribution has cumulative norm at most $CgE$ by
\eqref{eq:early-tree-score} and
\[
 \sum_j\eta_j^2\sourcenormT{\diff\bm\xi_j}^{\,2}
 \le4g^2E^2.
\]
At $t=1$, zero prehistory contributes the boundary impulse
$\bm r_i^{(1)}$, of source norm at most one.
We call these contributions \emph{stopped}: they are included
in the final bound without taking further differences.

\paragraph{A gain-direction distribution derivative creates an atom.}
For $e=(j,a,b)$ and $\bm v\in[0,1]^{m_j}$, define the local
centered factor
\[
 \chi^\lambda_{e,\bm v}(\bm y)
 =\frac{h_e(\bm y)}{q_e(\bm y)}
   \bigl(v_b-\ip{\bm\sigma_{j,a}(\bm y)}{\bm v}\bigr).
\]
The gain-direction score in \eqref{eq:tree-score} is a centered
sum of these factors over players, weighted by their learning
rates and source probabilities. Fix $t\ge2$ and $\theta\in[0,1]$.
Introduce an independent copy of the original tree sample and
a separate gain sample:
\[
 \bm T,\bm T'\sim\treeP_{\bm y_\theta^{(t)}},
 \qquad
 \bm S\sim\treeP_{\bm z^{(t)}},
\]
with all three samples independent. Write $\E_{t,\theta}^{+}$
for expectation under this enlarged distribution. The construction in
Appendix~\ref{app:early-closure} gives an order-zero \emph{atom}
$Z_\theta^{(0)}$, with player weights $\eta_j/g$, such that
\[
 \E_\theta\!\left[
   \bm G_i(\bm T)
   \nabla_{\eta\bm r^{(t)}}
   \ln\treeP_{\bm y_\theta^{(t)}}(\bm T)
 \right]
 =c_0\E_{t,\theta}^{+}\!\left[
   \bm G_i(\bm T)Z_\theta^{(0)}
 \right],
 \qquad c_0:=2Bg\sqrt\ell.
\]
Here $B$ is the universal constant in
\cref{lem:normalized-derivatives-full}. Indeed, averaging over
$\bm T'$ replaces $I_e(\bm T)-I_e(\bm T')$ by
$I_e(\bm T)-\E I_e$, and averaging over $\bm S$ represents the
gain direction by \eqref{eq:centered-tree-identity}. Thus the
equality holds after averaging over the auxiliary samples,
not pointwise between the original score and the atom.
We suppress the sample and score arguments of the atom.

Define the continuing sequence, for each fixed $\theta$, by
\[
 \bm U_\theta^{(t)}
 :=\begin{cases}
    \E_{t,\theta}^{+}
      [\bm G_i(\bm T)Z_\theta^{(0)}],&t\ge2,\\
    0,&t\le1.
   \end{cases}
\]
The first difference is therefore the exact sequence identity
\[
 \diff\bm r_i
 =\bm V+c_0\int_0^1\bm U_\theta\,d\theta.
\]
Here $\bm V$ contains the boundary impulse at $t=1$ and the
error-direction integral above for $t\ge2$.

\paragraph{The second difference adds a factor or raises its order.}
Fix one continuing sequence $\bm U=\bm U_\theta$, whose
integrand has the form $\bm GZ^{(0)}$, with
$\bm G=\bm G_i(\bm T)$. For $t\ge3$, interpolate all its score
arguments between times $t-1$ and $t$ and apply the product rule.
Every sampling distribution contributes a derivative, including the distributions
of the independent copy and the gain sample. A gain-direction
distribution derivative creates a new order-zero atom. A gain derivative
of the existing atom creates an order-one atom. Collecting the
error directions into $\bm V$ gives the exact identity
\[
 \begin{aligned}
  (\diff\bm U)^{(t)}
  ={}&\bm V^{(t)}
   \quad+\quad\underbrace{\int c_1(\nu)\,
      \E_{t,\nu,1}\!\left[
       \bm GZ_{\nu,1}^{(0)}Z_{\nu,2}^{(0)}
      \right]d\nu}_{\text{differentiate a sampling distribution}}
   \quad+\quad\underbrace{\int c_2(\nu)\,
      \E_{t,\nu,2}\!\left[
       \bm GZ_\nu^{(1)}
      \right]d\nu}_{\text{differentiate the existing atom}}.
 \end{aligned}
\]
The two integrals include all interpolation parameters and finite
branching choices for their respective derivative families.
Each $\E_{t,\nu,k}$ uses the complete auxiliary-sample distribution of
that branch, with the old observable $\bm G$ viewed as a function
of the enlarged sample tuple. The coefficients $c_k(\nu)$ are
independent of time. Sample and score arguments of the atoms
are suppressed. At $t=2$, set $\bm V^{(2)}=\bm U^{(2)}$ and
all continuing terms to zero. All terms are zero for $t\le1$.
As throughout the recursion, $\bm V$ denotes the stopped sequence
for the expression currently being differentiated.

More generally, a gain-direction distribution derivative adds $Z^{(0)}$,
whereas a gain derivative of $Z^{(r)}$ raises its order to $r+1$.
An order-$r$ atom contains normalized order-$r$ derivatives of
the local factor, together with auxiliary payoff samples.
\Cref{lem:normalized-derivatives-full} and the exact constructions
in Appendix~\ref{app:early-closure} justify these two operations.

\paragraph{The form preserved by the recursion.}
Every continuing branch at depth $p$ therefore has the form
\[
 \bm U^{(t)}
 =\one_{\{t>p\}}\E_t\!\left[
    \bm G(\Omega)\prod_{\nu=1}^h
    Z_\nu^{(r_\nu)}(\Omega,\bm y_\nu^{(t)})
  \right],
 \qquad
 \sum_{\nu=1}^h(r_\nu+1)\le p.
\]
Here $\E_t$ averages the auxiliary product-tree samples $\Omega$
under their time-$t$ distributions, and $\bm G$ is a fixed observable
with $\sourcenorm{\bm G}\le1$. Each continuing difference adds
one unit of degree. The indicator ensures positive represented
gain times and produces a boundary impulse at the next
difference. Initially there are no atoms, so the gain is a
depth-zero expression. Throughout, an \emph{early expression}
means the admissible construction of
Appendix~\ref{app:early-closure}, including its sample-independence
conditions and fixed convex-combination rules for the score
arguments. The next subsection bounds this recursion.

\subsection{Bounding the early expansion}
\label{sec:early-sketch}

During the first $\ell$ differences, we avoid the weighted
curvature energy $Y$. Instead, player centering keeps the
learning rates at the scale $(\sum_i\eta_i^2)^{1/2}\le g$.
We need two bounds: products of atoms remain normalized,
and the total weight of the continuing branches becomes small.

\paragraph{Products of centered atoms remain controlled.}
Conditional on its own payoff samples, each atom is a sum of
independent centered player contributions, with squared weights
summing to at most one. The local normalization and the fact
that a tree has at most one outgoing edge from the selected
source row bound each player's summand.
The normalization $2\sqrt\ell$ then gives, by
\cref{lem:centered-moment},
\[
 \norm{Z^{(r)}}_{L^{2\ell}}\le1,
 \qquad
 \E\prod_{\nu=1}^h|Z_\nu^{(r_\nu)}|^2\le1
 \quad(h\le\ell).
\]
The product bound uses H\"older, not independence between
atoms. Thus the integrand of every early expression has
$L^2$ norm at most one. These are sample-moment bounds at
each round, not yet bounds on the cumulative filtered norm.

\cref{lem:normalized-derivatives-full} bounds error derivatives
of existing factors, while \eqref{eq:early-tree-score} bounds
error derivatives of sampling distributions. Together with the two
continuing operations in Section~\ref{sec:prediction-setup},
they give the following exact closure, proved in
Appendix~\ref{app:early-closure}.

\begin{lemma}[Early one-difference closure]
\label{lem:early-closure}
Every depth-$p$ early expression, $p<\ell$, has an exact
decomposition
\[
 \diff\bm U=\bm V+\int c(\nu)\bm U_\nu\,d\nu,
\]
where $\bm V$ consists of the boundary impulse and
error-direction terms, every $\bm U_\nu$ has early depth $p+1$,
and the coefficients are time-independent. Moreover,
\[
 \sourcenormT{\bm V}\le1+C_E(p+1)gE,
 \qquad
 \int|c(\nu)|\,d\nu\le C_Eg\sqrt\ell\,(p+1).
\]
\end{lemma}
Here $\nu$ collects the interpolation parameters and finite
branching choices from the two derivative families above.
The factor $p+1$ accounts for the number of distributions and the total
order of existing factors. The extracted coefficients depend
on fixed rates and representation choices, not on time, so
they can be taken outside the time filters.

\paragraph{Filtering each exact decomposition.}
For a filter $\mathrm K=\sum_{s\ge0}c_s\shift^s$, write
$\kernorm{\mathrm K}=\sum_s|c_s|$. Backward shifts contract
the finite-prefix norm, so
\[
 \sourcenormT{\mathrm K[\bm W]}
 \le\kernorm{\mathrm K}\sourcenormT{\bm W}.
\]
The filter identity and \cref{lem:early-closure} give
\[
 \earlyfilter_p[\bm U]
 =\earlyfilter_{p+1}[\bm V]
  +\int c(\nu)\earlyfilter_{p+1}[\bm U_\nu]\,d\nu.
\]
We bound the first term using
$\kernorm{\earlyfilter_{p+1}}<48k^p$, proved in
\cref{lem:filter-full}, and take another difference of each
continuing expression in the second term.

Let $a_p$ be the total absolute coefficient mass of the
continuing branches at depth $p$, starting with $a_0=1$.
Each branch's coefficient is the product of its one-step
coefficients. The closure gives
$a_{p+1}\le C_Eg\sqrt\ell\,(p+1)a_p$.
Including the filter cost, for $p<\ell$,
\[
 a_{p+1}k^{p+1}
 \le C_Egk\sqrt\ell\,(p+1)\,a_pk^p
 \le\kappa\,a_pk^p,
\]
because $p+1\le k$, $\sqrt\ell\le\sqrt k$, and
$\kappa=C_Egk^{5/2}$. Hence
\begin{equation}\label{eq:early-mass}
 a_pk^p\le\kappa^p
 \qquad(0\le p\le\ell).
\end{equation}
For $\kappa\le1/16$, the total early stopped contribution
is at most
\[
 48\sum_{p=0}^{\ell-1}
   \kappa^p[1+C_E(p+1)gE]
 \le C(1+gE).
\]
At depth $\ell$, define
\[
 \latefilter_0:=\diff^3\filt_N^N,
 \qquad
 \earlyfilter_\ell=\res_\ell^\ell\latefilter_0.
\]
The remaining early resolvent has kernel mass $k^\ell$,
already included in $a_\ell k^\ell\le\kappa^\ell$.
Thus, for the continuing depth-$\ell$ expressions,
\begin{equation}\label{eq:early-stage-output}
 \sourcenormT{\errfilt[\bm r_i]}
 \le C(1+gE)
   +\kappa^\ell\sup_\nu
      \sourcenormT{\latefilter_0[\bm U_\nu]}.
\end{equation}
The early stage has bounded all stopped terms without using
$Y$. Only the late-filtered expressions remain, and their total
absolute weight, including the early filter cost, is at most
$\kappa^\ell$.

\subsection{Late stage: charge repeated transition indices to curvature}
\label{sec:late-sketch}

It remains to bound $\latefilter_0[\bm U]$ for each expression
passed from the early stage. We now allow stopped terms to use
$Y$. In particular, gain derivatives of existing factors stop
by \cref{lem:normalized-derivatives-full}, rather than raising
their orders. Only a residual part of a sampling-distribution derivative
can continue. The main idea is: a first occurrence of a
transition index stores a factor, and a repeated occurrence supplies
the second factor needed for a curvature bound.

\paragraph{Separate the curvature-controlled score from the residual.}
For $e=(i,a,b)$, \cref{lem:scalar-bounds} supplies a split
\[
 \frac{h_e}{q_e}=\gamma_e+\delta\alpha_e
\]
with
\begin{equation}\label{eq:residual-bounds}
 0<\alpha_e\le1,
 \qquad
 \gamma_e^2\le C_\gamma h_e/\beta_i,
 \qquad
 \delta\alpha_e^2\le C_\alpha\sqrt{h_e/\beta_i}.
\end{equation}
The first curvature bound controls the entire centered
$\gamma$-score, by \cref{lem:tree-score-full}:
\[
 \E\left|
   \sum_e(I_e-\E I_e)\gamma_e\widetilde w_e
 \right|^2
 \le\sum_{i,a}\omega_i
      \quadform_{\bm y_{i,a}}(\bm w_{i,a}).
\]
After the score comparison described below, this term stops
using $Y$. Only the $\delta\alpha$ residual is split by
transition index. Its useful bound requires two factors
of $\alpha_e$, whereas one distribution derivative supplies only one.

For a fresh index $e=(i,a,b)$ and a payoff vector
$\bm v\in[0,1]^{m_i}$, retain the bounded centered factor
\[
 \chi^\alpha_{e,\bm v}(\bm y)
 =\alpha_e(\bm y)
   \bigl(v_b-\ip{\bm\sigma_{i,a}(\bm y)}{\bm v}\bigr),
 \qquad
 |\chi^\alpha_{e,\bm v}(\bm y)|
 \le\alpha_e(\bm y)\le1.
\]
A depth-$j$ late expression keeps the early atoms and adds
$j$ such factors, whose recorded indices are distinct wherever
the observable is nonzero. The gate becomes
$\one_{\{t>\ell+j\}}$, and the early atom orders do not increase.
The extra factors preserve the integrand's $L^2$ bound.
The full admissible class, including sample-dependent index
selection, is defined in Appendix~\ref{app:late-prediction}.
A depth-$\ell$ early expression is a depth-zero late expression.

\paragraph{A repeated index supplies the second factor and stops.}
Suppose a residual distribution derivative selects an index $e=(i,a,b)$
that is already recorded. Its old $\chi^\alpha$ factor supplies
a second $\alpha_e$, and its centered payoff has magnitude at most
one. At a common trial score $\bm y$, suppress the represented
gain time and write $\widetilde r_e$ for the row-centered gain
from \eqref{eq:row-geometry}. Then
\[
 \begin{aligned}
  \delta\eta_i\alpha_e^2|\widetilde r_e|
  &\le C\beta_i^{-1/2}\eta_i\sqrt{h_e}
           |\widetilde r_e|\\
  &\le C\beta_i^{-1/2}\eta_i
       \sqrt{\quadform_{\bm y_{i,a}}(\bm r_{i,a})}.
 \end{aligned}
\]
The second inequality retains one summand of the row-curvature
quadratic form. \Cref{lem:late-window} compares the different
trial scores and transfers this bound to the represented gain
time, at cost $C2^j$. Since $\omega_i\ge\beta_i^{-1}$, each
repetition contribution has cumulative cost at most $C2^jY$.
The bound is uniform in the selected index and is applied
before introducing a new gain sample. Removing the old
factor used in the bound leaves an $L^1$-bounded product.

The same score comparison controls the $\gamma$-score and
gain derivatives of existing atoms and residual factors.
Thus these also stop using $Y$, while error directions
continue to stop using $E$.

\paragraph{A fresh index is recorded and continues.}
For a fresh index, introduce an independent copy and a separate
gain sample as in Section~\ref{sec:prediction-setup}, and split
by player and sign. The gain root selects one source row $a$, and
the indicated inclusion tree has at most one outgoing edge
from that row:
\[
 \sum_b I_{i,a,b}(\bm T_i)\le1.
\]
Hence the sum over destinations selects at most one edge
rather than costing a factor $m_i$. Append its $\chi^\alpha$
factor and record its index. Freshness is imposed by a bounded
indicator in the observable, not by conditioning the tree distribution.
Each such continuation has coefficient $\delta\eta_{\max}$
in absolute value. Appendix~\ref{app:late-prediction} gives
the exact construction and proves the following closure.

\begin{lemma}[Late one-difference closure]
\label{lem:late-closure}
Every depth-$j$ late expression, $0\le j\le N$, has an exact
decomposition
\[
 \diff\bm U=\bm V+\int c(\nu)\bm U_\nu\,d\nu,
\]
with time-independent coefficients. The stopped sequence
$\bm V$ contains the boundary impulse, error directions,
the $\gamma$-score, derivatives of existing factors, and
repeated transition indices. For a universal $C_L$,
\[
 \sourcenormT{\bm V}\le C_Lk^2 8^j(1+gE+Y),
 \qquad
 \int|c(\nu)|\,d\nu
 \le2n\delta\eta_{\max}(1+2\ell+2j).
\]
If $j<N$, every continuing expression has late depth $j+1$.
At $j=N$, the continuing term is absent.
\end{lemma}
The continuation bound counts at most $1+2\ell+2j$ sampling
distributions, $n$ player choices, and two signs. Every continuation
adds one fresh index. Since there are only $N$ indices,
the expansion leaves only stopped terms at depths
$0,\ldots,N$, including the boundary contribution at depth $N$.

\paragraph{The finite expansion has bounded total cost.}
To apply the remaining differences, define
\[
 \latefilter_j:=\res_N^N\diff^{N+3-j}
 \quad(0\le j\le N+1).
\]
This agrees with $\latefilter_0$ above and gives
$\latefilter_j=\latefilter_{j+1}\diff$.
By \cref{lem:filter-full},
$\kernorm{\latefilter_{j+1}}<24(N+1)^j$.
Set
\[
 a_{\mathrm{late}}
 :=2n\delta\eta_{\max}(1+2\ell+2N).
\]
Filtering and summing the stopped contributions yields
\begin{equation}\label{eq:late-stage-output}
 \begin{aligned}
  \sourcenormT{\latefilter_0[\bm U]}
  &\le24C_Lk^2(1+gE+Y)
       \sum_{j=0}^N[8(N+1)a_{\mathrm{late}}]^j\\
  &\le25C_Lk^2(1+gE+Y).
 \end{aligned}
\end{equation}
The public choice of $\delta$ makes
$8(N+1)a_{\mathrm{late}}<1/25$, as verified in
Appendix~\ref{app:filters}. Thus $\delta$ offsets both
the $8^j$ stopping cost and the late filter mass.
The sum is finite, includes every boundary impulse, and has
no terminal remainder. Because the filters contain only backward
shifts, evaluating them through round \(T\) introduces no quantities
with rounds after \(T\). Hence the same
\(E\) and \(Y\) control every term.

\subsection{Combining the two stages proves the prediction bound}
\label{sec:filter-assembly}

Combining \eqref{eq:early-stage-output} and
\eqref{eq:late-stage-output} with the effective-error identity
gives
\[
 \begin{aligned}
  E
  &\le1+\max_i\sourcenormT{\errfilt[\bm r_i]}\\
  &\le1+C(1+gE)
       +25C_Lk^2\kappa^\ell(1+gE+Y)\\
  &\le C_0+C_1gE+C_Pk^2\kappa^\ell Y,
 \end{aligned}
\]
where the last step uses that $k^2\kappa^\ell$ is uniformly
bounded for $k=\ell+1$ and $\kappa\le1/16$. This proves
\cref{lem:prediction}. The absorption argument in
Section~\ref{sec:two-estimates}, together with \cref{lem:energy},
then completes the proof of \cref{thm:base}.

\section{Adversarial robustness via a common-prefix wrapper}\label{sec:robustness}

In this section, we present a black-box wrapper to turn self-play swap regret bounds into adversarial swap regret bounds. It applies whenever a simultaneous run of the base dynamics admits an anytime self-play regret bound that can be absorbed into the adversarial square-root envelope.
For a base run beginning after a common prefix of length $W$, we write
\begin{equation}\label{eq:switching-monitor}
 R_{i,\mathrm{base}}(s)
 =\sum_{a=1}^{m_i}\max_{b\in[m_i]}\sum_{u=1}^s r_{i,a,b}^{(W+u)},
 \qquad s\ge0.
\end{equation}

\begin{algorithm}[H]
\caption{Common-prefix robust wrapper for player $i$}\label{alg:robust}
\small
\textbf{Input.} Base dynamics, a threshold $B_i\ge0$, and a common prefix length $W\ge1$.

\textbf{Anytime fallback $\mathsf{FB}_i$.} Use epochs of lengths $M=1,2,4,\ldots$. At each epoch start, initialize $q_{a,b}=1/m_i$ and set $\eta_{\mathrm{fb}}=\sqrt{A_i/M}$. In each round, play the stationary distribution $\bm x$ of $\bm Q$, observe $\bm v$, set $g_{a,b}=x_av_b$, and update
\[
 q_{a,b}\leftarrow
 \frac{q_{a,b}(1+\eta_{\mathrm{fb}}g_{a,b})}
 {1+\eta_{\mathrm{fb}}\ip{\bm q_a}{\bm g_a}}.
\]
After $M$ rounds, double $M$ and reset the rows uniformly. This fallback may be stopped after any $s$ rounds and satisfies $\Reg_{i,\mathsf{FB}}(s)\le4\sqrt{A_i s}$.

\textbf{Switching rule.}
\begin{algorithmic}[1]
 \State Run a new copy of $\mathsf{FB}_i$ for rounds $1,\ldots,W$.
 \State Discard it and initialize the base dynamics from its prescribed initial state.
 \State After each base round $s$, compute \eqref{eq:switching-monitor}. At the first $s$ with $R_{i,\mathrm{base}}(s)>B_i$, discard the base state and run a new copy of $\mathsf{FB}_i$ forever from the next round.
\end{algorithmic}
\end{algorithm}

\begin{lemma}[Common-prefix wrapper]\label{lem:common-prefix-wrapper}
Suppose that, in any fixed game, a simultaneous run of the base dynamics satisfies the anytime self-play bound $R_{i,\mathrm{base}}(s)\le B_i$ for every player $i$ and $s\ge0$. Then \cref{alg:robust} never makes a threshold-triggered switch in self-play. Against arbitrary payoff-vector sequences, for $T>W$,
\begin{equation}\label{eq:generic-wrapper-bound}
 \Reg_i(T)\le4\sqrt{A_iW}+(B_i+1)+4\sqrt{A_i(T-W)}.
\end{equation}
Consequently, if $B_i+1\le\gamma_i\sqrt{A_iW}$ for some $\gamma_i\ge0$, then, for every $T$,
\begin{equation}\label{eq:generic-wrapper-envelope}
 \Reg_i(T)\le\sqrt{(4+\gamma_i)^2+16}\,\sqrt{A_iT}.
\end{equation}
\end{lemma}
The proof is in Appendix~\ref{app:switching-rule}. The same proof applies to a nondecreasing anytime self-play bound $R_{i,\mathrm{base}}(s)\le b_i(s)$: use the threshold $b_i(s)$ whenever $b_i(s)+1\le\gamma_i\sqrt{A_i(W+s)}$ for all $s\ge1$.

For our base dynamics and default rates, set
\begin{equation}\label{eq:switching-scales}
 B_i=\frac4c k^{5/2}\sqrt{A_iA},
 \qquad
 W=\left\lceil\frac{9k^5A}{c^2}\right\rceil.
\end{equation}
By \eqref{eq:generic-bound}, the base dynamics satisfy the anytime self-play bound
$R_{i,\mathrm{base}}(s)\le\frac3c k^{5/2}\sqrt{A_iA}=\frac34B_i$.
Appendix~\ref{app:robustness} verifies
\begin{equation}\label{eq:switching-specialization}
 B_i+1\le\frac53\sqrt{A_iW},
 \qquad
 4\sqrt{A_iW}+\frac34B_i<4B_i.
\end{equation}
Thus \cref{lem:common-prefix-wrapper} with $\gamma_i=5/3$ gives $\Reg_i(T)<7\sqrt{A_iT}$ against arbitrary payoff sequences. In self-play no switch occurs, and \eqref{eq:switching-specialization} gives
$\Reg_i(T)<4B_i=(16/c)k^{5/2}\sqrt{A_iA}$.
This proves \cref{thm:robust}.
\FloatBarrier

\section*{Acknowledgments}
GPT-5.6-Sol and GPT-6-Astra were used to assist with the development of this paper.

\bibliography{paper}

\appendix
\section{Supplementary background, proof roadmap, and parameters}\label{app:context}
The main text gives the proof mechanism. The table below locates the complete verification of each component.
\begin{center}
\scriptsize
\setlength{\tabcolsep}{4pt}
\begin{tabularx}{\textwidth}{@{}>{\raggedright\arraybackslash}p{.30\textwidth}>{\raggedright\arraybackslash}p{.23\textwidth}>{\raggedright\arraybackslash}X@{}}
\toprule
Component & Complete verification & Output used next\\
\midrule
Public parameters and absorption & Appendices~\ref{app:parameters}--\ref{app:constant-selection} & Conditions used in Section~\ref{sec:two-estimates}\\
Row normalization and curvature comparison & Appendix~\ref{app:scalar-potential} & Bounds used in Sections~\ref{sec:energy-idea} and \ref{sec:late-sketch}\\
Uniform normalized derivatives & Appendix~\ref{app:normalized-derivatives} & Derivative bounds used in both closures\\
Stationary-tree representation and moments & Appendix~\ref{app:trees-moments} & Inputs to the early and late closures\\
Early player-centered closure & Appendix~\ref{app:early-closure} & \Cref{lem:early-closure} and the $\kappa^\ell$ handoff\\
Late transition-index closure & Appendix~\ref{app:late-prediction} & \Cref{lem:late-closure}\\
Two-scale filters and assembly & Appendix~\ref{app:filters} & \Cref{lem:prediction}\\
Normalization solver and implementation & Appendix~\ref{app:solver} & Effective-score interface and complexity\\
Fallback and common-prefix wrapper & Appendix~\ref{app:fallback} & \Cref{lem:common-prefix-wrapper,thm:robust}\\
\bottomrule
\end{tabularx}
\end{center}

\subsection{Correlated-equilibrium consequence}\label{app:ce}
Let
\[
 \pi_T=\frac1T\sum_{t=1}^T\bigotimes_i\bm x_i^{(t)}
\]
be the time-averaged product distribution in a fixed game. For every player $i$ and deviation map $\varphi:[m_i]\to[m_i]$, multilinearity gives
\[
 \E_{\bm S\sim\pi_T}\!\left[
 u_i(\varphi(S_i),\bm S_{-i})-u_i(\bm S)
 \right]
 =\frac1T\sum_{t=1}^T\sum_a x_{i,a}^{(t)}
 \left(v_{i,\varphi(a)}^{(t)}-v_{i,a}^{(t)}\right)
 \le\frac{\Reg_i(T)}T.
\]
Let $B_i=(4/c)k^{5/2}\sqrt{A_iA}$ as in \cref{sec:robustness}. Then $\pi_T$ is a $((\max_i B_i)/T)$-approximate correlated equilibrium under the base dynamics and a $(4\max_i B_i/T)$-approximate correlated equilibrium when all players use the robust dynamics in self-play. Note that $\pi_T$ is a time-averaged product distribution, not a last iterate or the product of the players' average strategies \citep[Theorem~3]{blum2007}.

For the single-action convention
\[
 g_i(a,b,\pi)=\E_{\bm S\sim\pi}\left[
 \one_{\{S_i=a\}}\bigl(u_i(b,\bm S_{-i})-u_i(\bm S)\bigr)\right],
\]
one has $0\le\max_{a,b}g_i(a,b,\pi)\le\sum_a\max_bg_i(a,b,\pi)\le m_i\max_{a,b}g_i(a,b,\pi)$. Thus positive tolerances may differ by a factor $m_i$, although the zero-tolerance conditions agree.

\subsection{Public parameters}\label{app:parameters}
In addition to \eqref{eq:public-scales}, set
\begin{equation*}
 \delta=\frac1{n(N+\ell+2)^2},
 \qquad
 J=\lceil\log(1/\delta)\rceil,
 \qquad
 \beta_i=\frac{\log m_i}{2^{10}m_iJ}.
\end{equation*}
Let $d_i$ be the least power of two at least $\lceil\log(1/\beta_i)\rceil$. For the default allocation, use \eqref{eq:default-rates}. With the proof-only constant $C_\omega$ fixed in Appendix~\ref{app:constant-selection}, define
\begin{equation*}
 \omega_i=C_\omega\frac{m_i-1}{\beta_i},
 \qquad
 V=\sum_i\omega_iA_i.
\end{equation*}
Then $\eta_{\max}\le g$, $m_i\beta_i\le2^{-10}$, $0<\delta<1/2$, and $\log m_i\le J$. All parameters depend only on public action counts and universal constants.

\subsection{Noncircular selection of universal constants}\label{app:constant-selection}
First fix integer upper bounds at least one for all universal constants in Appendices~\ref{app:scalar-potential}--\ref{app:filters}, and set $C_\omega=\max\{1,C_\gamma\}$. Fix
\begin{equation*}
 c=\min\left\{\frac1{16C_E},\frac1{4(C_1+1)},\frac1{2(C_0+1)}\right\}.
\end{equation*}
Then $\kappa=C_Ec\le1/16$. To choose $\ell_0$, write
\[
 v=\lceil\log(2+n+N)\rceil,
 \qquad k=\ell_0+v+1.
\]
Every nontrivial instance has $N\ge4$, so $v\ge3$ and $k\ge5$. Moreover,
\begin{equation*}
 N+\ell+2\le2^v(k+2),
 \qquad
 J\le1+3v+2\log(k+2)\le5k,
 \qquad
 V=2^{10}C_\omega J\sum_im_i^2(m_i-1)
 \le2^{13}C_\omega kN^{3/2}.
\end{equation*}
Choose $\ell_0$ large enough that
\begin{equation*}
 C_P\sqrt{2^{14}C_\omega}\,(\ell_0+1)^{5/2}16^{-\ell_0}\le1.
\end{equation*}
Since $k\le(\ell_0+1)(v+1)$ and $N\le2^v$,
\begin{align*}
 \tau_\ell\sqrt{2V}
 &\le C_P\sqrt{2^{14}C_\omega}\,(\ell_0+1)^{5/2}16^{-\ell_0}
 (v+1)^{5/2}2^{-13v/4}\le1.
\end{align*}
The final dimension factor starts at one and has consecutive ratio at most $2^{5/2-13/4}<1$, proving \eqref{eq:parameter-absorption} without changing any previously fixed lemma constant. Section~\ref{sec:two-estimates} then gives $E\le E_0=2(C_0+1)$ and \eqref{eq:generic-bound}. Because $cE_0\le1$, the default rates satisfy $\eta_iE_0\le1\le\sqrt{A_i}$. This proves \cref{thm:base} with the slack used by the robust wrapper.

\subsection{Optional unequal-action rate allocation}\label{app:heterogeneous-comparison}
Set
\begin{equation}\label{eq:heterogeneous-rates}
 U=\sum_jA_j^2,
 \qquad
 \eta_i=\frac{gA_i}{\sqrt U}.
\end{equation}
These rates satisfy \eqref{eq:rate-budget} and $\eta_i\le g$, so \eqref{eq:generic-bound} gives
\[
 \max_i\Reg_i(T)\le \frac4c k^{5/2}\sqrt U.
\]
The allocation minimizes $\max_iA_i/\eta_i$ subject to \eqref{eq:rate-budget}: if this maximum is $M$, then $\eta_i\ge A_i/M$ and hence $g^2\ge U/M^2$. Relative to the default worst-player bound, the improvement factor is
\[
 \sqrt{\frac{A_{\max}A}{U}}
 \le\sqrt{\frac{\sqrt n+1}{2}},
 \qquad A_{\max}=\max_iA_i.
\]
The inequality follows from Cauchy--Schwarz and is tight in order: $n-1$ binary players and one player with $m_i\log m_i=\Theta(\sqrt n)$ give a factor $\Theta(n^{1/4})$. Equal action counts recover the default rates. For the common-prefix wrapper, the relevant objective is $\max_iA_i/\eta_i^2$. Its minimum under \eqref{eq:rate-budget} is $A/g^2$, attained by the default square-root rates. Thus the optional allocation improves the base worst-player bound but not the present wrapper's worst-player prefix order.

\subsection{Technical comparison with related work}\label{app:related}

This appendix expands Section~\ref{sec:related-work}, focusing on the additional bounds needed to combine BM, tree representations, optimistic potential bounds, higher-order prediction, and finite-label recursions to obtain constant swap regret with the stated dependence on $n$ and $m$. Table~\ref{tab:comparison} compares selected bounds under full-information feedback, not every algorithm in the cited papers or methods with equal computational budgets.

\paragraph{Prediction, movement, and adaptive pacing.}

The potential argument follows the work on predictable sequences and regret bounded by variation in utilities (RVU) \citep{rakhlin2013,syrgkanis2015fast}. It balances prediction error against a negative movement or curvature term supplied by the regularizer. \citet{farina2022convexgames} obtain logarithmic external regret in convex games by lifting the action space and using self-concordant regularization. \citet{soleymani2025faster} instead control the learning rate of OMWU through Cautious Optimism. Their broader treatment gives a meta-algorithm that adjusts the pace of an FTRL learner \citep{soleymani2025cautious}. These external-regret methods do not directly establish a prediction bound for BM row gains, which depend on the player's entire stationary distribution.

\paragraph{Higher-order smoothness through the Blum--Mansour fixed point.}

The closest earlier analysis of the same fixed-point difficulty is the BM--OMWU result of \citet{anagnostides2022ce}. Their Theorem~2.4 states the Markov-chain tree formula, and Theorem~3.1 represents SL--OMWU using directed trees. For BM--OMWU, Section~4 and Lemmas~D.5--D.6 combine this formula with multivariate Cauchy bounds in logarithmic transition coordinates. Building on \citet{daskalakis2021}, they bound higher-order finite differences through the stationary-distribution operation and obtain polylogarithmic swap regret. Thus higher-order smoothness through the fixed point is an existing ingredient, distinct from an explicit higher-order predictor, since the cited dynamics use OMWU.

Appendix~\ref{app:normalized-derivatives} uses related ideas from complex analysis, but applies them to local centered factors rather than proving one global smoothness bound for the full stationary map. Normalizing the weights of rooted trees and differentiating their distribution gives the centered edge score in \eqref{eq:tree-score}. Neither this normalization nor the centered-score identity is a separate novelty. Lemma~\ref{lem:normalized-derivatives-full} additionally gives every normalized derivative both an ordinary bound and a bound in terms of the same real row curvature used by the potential argument. The first bound allows the early stage to raise derivative orders without using the energy weights that depend on the game dimensions. The second allows the late stage to bound those derivatives directly by curvature.

\paragraph{Movement-based control of the stationary strategy.}
\citet{anagnostides2022swap} combine Blum--Mansour with optimistic FTRL and a log-barrier. The barrier supplies a negative movement term. Their Lemma~4.2 uses the tree theorem to transfer multiplicative row stability to a bound on movement of the stationary strategy. Their Theorem~4.4 uses nonnegative swap regret to bound second-order path length. This gives $O(\log T)$ swap regret. \citet{tsuchiya2026} uses a hybrid entropy/log-barrier regularizer: entropy controls optimistic prediction error, the log-barrier controls row movement, and a global sensitivity theorem for rooted trees transfers this control to the stationary distribution without requiring local closeness.

Our proof uses a different interface rather than a sharper sensitivity inequality. The row potential has Hessian $\quadform_z$, and stationarity cancels the linear term after summing over source rows, giving the nonnegative curvature energy in \eqref{eq:energy}. We spend this energy inside derivatives of the tree distribution instead of first bounding $\norm{x_i^{(t)}-x_i^{(t-1)}}$. This distinction is essential for the action dependence: applying a bound on stationary sensitivity or tree variance at every high-order step would introduce the weights $\omega_i=\Theta((m_i-1)/\beta_i)$ from the beginning. Our analysis delays this weighted curvature bound until the total
absolute weight of the continuing terms, including the early filter
cost, is exponentially small. 

\paragraph{Constant external regret and finite label recursions.}
The horizon-independent prediction mechanism is closest to the concurrent external-regret results of \citet{liu2026,abbadi2026}. ECHO-OFTRL of \citet{liu2026} preconditions a high-order difference by an EMA-resolvent cascade and recursively records exposed player labels. HOOD of \citet{abbadi2026} uses a discounted finite recurrence and an explicit expansion in which a repeated player is charged to movement while only a fresh player continues. Our filters are closest algorithmically to ECHO, and our exact stop/continue bookkeeping is closest to HOOD.

The BM reduction prevents a direct application of these self-play theorems. The original payoff vector is multiaffine in the opponents' strategies. Centered external gains may also depend on the player's own strategy, as in \citet{liu2026}, while remaining multiaffine in the original player strategies. In contrast, the row gains
\[
 r_{i,a,b}=x_{i,a}(v_{i,b}-v_{i,a})
\]
combine this payoff structure with the stationary-distribution map of the entire transition matrix. All rows of the same player are therefore coupled. BM transfers regret guarantees for the actual row feedback, but does not ensure that it satisfies the prediction bounds of an external-regret self-play theorem.

For a concrete illustration, let
\[
 Q=\begin{pmatrix}1-p&p\\ q&1-q\end{pmatrix},
 \qquad 0<p,q<1.
\]
Its stationary mass is $x_1=q/(p+q)$, so for a fixed payoff vector the row gain is $r_{1,2}=q(v_2-v_1)/(p+q)$. This is already rational, rather than affine, in either row parameter. For fixed $q$ and $v_2\ne v_1$, its derivatives with respect to $p$ need not vanish at any finite order. Thus treating BM rows as additional players does not preserve the original finite-game payoff structure. This identifies a missing hypothesis for a direct transfer, not an impossibility result for other reductions.

Our two closures separate these effects. During the first $\ell=O(\log(nm))$ differences, the differentiated product-tree distribution yields centered sums independent across players. Moment bounds combine the rates through $(\sum_i\eta_i^2)^{1/2}$ and keep the action-weighted energy out of the recursion. After the resolvent mass is included, the total absolute weight of the
continuing terms is at most $\kappa^\ell$. Only then does the proof refine its bookkeeping from players to transition indices $e=(i,a,b)$. The split
\[
 \frac{h(q)}q=\gamma(q)+\delta\alpha(q)
\]
separates a component paid collectively by curvature from a residual carrying an individual transition index. If the transition index repeats, the old $\alpha$ factor supplies the second factor needed for \eqref{eq:residual-bounds}, so the branch is charged to the energy and stops. Only a fresh index continues. Since there are $N=\sum_i m_i^2$ indices, the late expansion terminates exactly. The root selects a single source row, and a tree has at most one outgoing edge from that row. These properties avoid additional factors from the action counts. 
The two EMA cascades implement these two depths.

\paragraph{Why the row normalization is not interchangeable.}
The construction of $f_i$ must support the low initial potential in Lemma~\ref{lem:potential-full}, the two derivative bounds in Lemma~\ref{lem:normalized-derivatives-full}, and the residual split in Lemma~\ref{lem:scalar-bounds}. For example, the exponential response $f(s)=e^s$ has $h(q)=q$ and $h(q)/q=1$. With the present $0<\delta<1/2$ and $0<\alpha\le1$, a split $1=\gamma+\delta\alpha$ forces $\gamma>1/2$. For fixed $\beta>0$, this cannot satisfy $\gamma^2\le C_\gamma q/\beta$ uniformly as $q\downarrow0$. Ordinary entropy normalization therefore cannot simply replace the constructed map in this proof. This observation concerns the specific curvature split, not the possibility of other entropy-based constant-regret algorithms.

\paragraph{Technical contribution and scope.}

The main technical contribution is the prediction bound in \cref{lem:prediction}, which combines the derivative bounds, the bound for repeated transitions, and the delayed use of curvature to yield constant swap regret. Relative to the hybrid BM bound $O(nm^2\sqrt{\log m\log T})$ of \citet{tsuchiya2026} displayed in Table~\ref{tab:comparison}, our bound removes the horizon dependence and improves the simultaneous polynomial dependence on $n$ and $m$. This is not a uniform comparison over every regime: \citet{chen2020} has a smaller exponent of $m$, while \citet{hait2025} obtain horizon-independent, comparator-adaptive $\Phi$-regret in games with nonnegative social external regret. Our arbitrary-game guarantee does not include their comparator-complexity refinement. The constant-regret results of \citet{liu2026,abbadi2026} have polylogarithmic action dependence but concern the weaker external-regret/CCE benchmark.

\section{Row normalization, potential, and curvature comparison}\label{app:scalar-potential}
This appendix proves the scalar residual bounds and the potential/curvature facts used in Sections~\ref{sec:energy-idea} and \ref{sec:late-sketch}. Fix a player $i$ and suppress its index on the scalar functions, while retaining the local dimension $m_i$. Write $\beta=\beta_i$ and $d=d_i$. The public parameters satisfy $m_i\ge2$, $0<\delta\le1/2$, $m_i\beta\le1/(4e)$, and $d\ge\log(1/\beta)$. In particular $d\ge5$.

\begin{lemma}[Scalar bounds]\label{lem:scalar-bounds}
For $q=f(s)$, define
\begin{equation*}
 h(q)=f'(s),
 \qquad
 \lambda(s)=\frac{f'(s)}{f(s)},
 \qquad
 \alpha(q)=\frac1{1+\delta u(s)/2},
 \qquad
 \gamma(q)=\frac{h(q)}q-\delta\alpha(q).
\end{equation*}
There are universal constants $C_\gamma,C_\alpha$ such that, for $0<q\le1$,
\begin{equation*}
 0<h(q)\le q,
 \qquad |h'(q)|\le2,
 \qquad 0<\alpha(q)\le1,
 \qquad
 \gamma(q)^2\le C_\gamma h(q)/\beta,
 \qquad
 \delta\alpha(q)^2\le C_\alpha\sqrt{h(q)/\beta}.
\end{equation*}
Moreover,
\begin{equation}\label{eq:alpha-and-row-curvature}
 0\le\frac d{ds}\alpha(f(s))\le\frac\delta2\alpha(f(s))^2,
 \qquad
 \sum_bh(q_b)\ge\frac3{10}
\end{equation}
for every normalized row. If $s\ge0$ and $f(s)\le1$, then $h(f(s))\ge(2/5)f(s)$.
\end{lemma}

\begin{proof}
Write $u=u(s)>0$ and $B(u)=1+u+\delta u^2/4$. Since $s=u^{-1}-u$,
\[
 u'=-\frac{u^2}{1+u^2},
 \qquad
 B(u)=(1+au)(1+bu),
 \qquad
 a=\frac{1+\sqrt{1-\delta}}2,
 \quad b=\frac{1-\sqrt{1-\delta}}2.
\]
Direct logarithmic differentiation of \eqref{eq:scalar-response} gives
\begin{align}
 \lambda(s)
 &=\frac{du}{(1+u^2)(1+du)}
   +\frac{au^2}{(1+u^2)(1+au)}
   +\frac{bu^2}{(1+u^2)(1+bu)}\notag\\
 &=\frac1{1+u^2}\frac{du}{1+du}
   +\frac{u^2}{1+u^2}\frac{1+\delta u/2}{B(u)}.
 \label{eq:lambda-u}
\end{align}
The last line is a convex combination of numbers in $(0,1]$. Hence $0<\lambda\le1$, $f$ is strictly increasing, and $0<h(q)=q\lambda\le q$.

Each of the three positive terms in the first line of \eqref{eq:lambda-u} has logarithmic derivative of magnitude at most one. Indeed $(\ln u)'=-u/(1+u^2)$. For the first term, the remaining bracket is
\[
 1-\frac{2u^2}{1+u^2}-\frac{du}{1+du},
\]
for the terms with $c=a,b$, it is
\[
 \frac2{1+u^2}-\frac{cu}{1+cu}.
\]
Both brackets have magnitude at most two, while $u/(1+u^2)\le1/2$. Thus $|\lambda'/\lambda|\le1$, and
\[
 |h'(f(s))|=\left|\frac{f''(s)}{f'(s)}\right|
 =\left|\lambda+\frac{\lambda'}\lambda\right|\le2.
\]

\paragraph{Positive-score region.}
For $u\le1$, Bernoulli's inequality gives $(1+1/(du))^d\ge1+1/u$. Retaining the first term of \eqref{eq:lambda-u},
\begin{equation}\label{eq:positive-h-lower}
 \frac{h(f(s))}{\beta}
 \ge\frac{d(1+u)}{(1+du)(1+u^2)B(u)}\ge\frac29.
\end{equation}
If also $f(s)\le1$, then
\[
 1+\frac1{du}\le\left(\frac{B(u)}\beta\right)^{1/d}
 \le2(9/4)^{1/d}\le\frac52.
\]
Therefore $du/(1+du)\ge2/5$, while $(1+\delta u/2)/B(u)\ge1/2$, and the second line of \eqref{eq:lambda-u} gives $\lambda\ge2/5$. For $u\ge1$, $f(s)\le e\beta$. At most $e m_i\beta\le1/4$ of a normalized row's probability mass lies in this region. On the remaining mass, $h(q)\ge(2/5)q$, so $\sum_bh(q_b)\ge(3/4)(2/5)=3/10$.

\paragraph{Negative-score region and residual split.}
For $u\ge1$, direct bounds in \eqref{eq:lambda-u} give universal two-sided comparisons
\begin{equation}\label{eq:negative-comparisons}
 \frac{f(s)}\beta\asymp\frac1{u(1+\delta u)},
 \qquad
 \lambda(s)\asymp\frac1u,
 \qquad
 \frac{h(f(s))}\beta\asymp\frac1{u^2(1+\delta u)}.
\end{equation}
Indeed, $1\le(1+1/(du))^d\le e$, $u(1+\delta u)/4\le B(u)\le2u(1+\delta u)$, and $a\ge1/2$. The cancellation
\[
 \frac{1+\delta u/2}{B(u)}-\frac\delta{1+\delta u/2}
 =\frac{1-\delta}{B(u)(1+\delta u/2)}
\]
implies
\[
 |\gamma(f(s))|\le C\bigl[u^{-2}+u^{-1}(1+\delta u)^{-2}\bigr].
\]
Its square is at most $Ch(f(s))/\beta$ by \eqref{eq:negative-comparisons}. Also
\[
 \frac{\delta\alpha^2}{\sqrt{h/\beta}}
 \le C\frac{\delta u}{(1+\delta u)^{3/2}}\le C.
\]
For $u\le1$, both weighted bounds follow from \eqref{eq:positive-h-lower}, $\lambda\le1$, and $\delta\le1/2$. Finally,
\begin{equation*}
 \frac d{ds}\alpha(f(s))
 =\frac\delta2\alpha(f(s))^2\frac{u^2}{1+u^2},
\end{equation*}
which proves \eqref{eq:alpha-and-row-curvature}.
\end{proof}

\subsection{Potential and local comparison}
Define $P,c_\Phi,\Phi$ as in \eqref{eq:potential}. The primitive is finite because $f(s)$ decays quadratically as $s\to-\infty$. The minimum in $c_\Phi$ occurs where $f(s)=1$.

\begin{lemma}[Potential and comparison]\label{lem:potential-full}
The potential is smooth and convex, and
\begin{equation*}
 \nabla\Phi(\bm z)=\bm q(\bm z),
 \qquad
 \Phi(\bm z)\ge\max_bz_b,
 \qquad
 0\le\Phi(0)\le2\log m_i.
\end{equation*}
For the quantities in \eqref{eq:row-geometry},
\begin{equation}\label{eq:hessian-full}
 \nabla_{\bm w}q_b=h_b\widetilde w_b,
 \qquad
 \bm w^\top\nabla^2\Phi(\bm z)\bm w
 =\quadform_{\bm z}(\bm w)
 =\min_{c\in\R}\sum_bh_b(w_b-c)^2
 \le\norm{\bm w}_\infty^2.
\end{equation}
If $\norm{\bm y-\bm z}_\infty\le\epsilon$, then
\begin{equation}\label{eq:local-comparison-full}
 e^{-4\epsilon}h_b(\bm z)\le h_b(\bm y)\le e^{4\epsilon}h_b(\bm z),
 \qquad
 e^{-\delta\epsilon}\alpha_b(\bm z)\le\alpha_b(\bm y)\le e^{\delta\epsilon}\alpha_b(\bm z),
\end{equation}
\begin{equation}\label{eq:quadratic-comparison-full}
 e^{-4\epsilon}\quadform_{\bm z}(\bm w)
 \le\quadform_{\bm y}(\bm w)
 \le e^{4\epsilon}\quadform_{\bm z}(\bm w).
\end{equation}
\end{lemma}

\begin{proof}
The lower bound $\sum_bh_b\ge3/10$ and the implicit-function theorem give a smooth normalizer. Differentiating $\sum_bf(z_b-\mu)=1$ gives
\[
 \nabla_{\bm w}\mu=\ip{\bm\sigma}{\bm w}.
\]
The terms involving \(\nabla_{\bm w}\mu\) in \(\nabla_{\bm w}\Phi\) cancel because $\sum_bq_b=1$, proving the gradient and Hessian formulas. Completing the square gives the minimum representation. Choosing $c=0$ and using $h_b\le q_b$ proves the upper bound. For $s_b=z_b-\mu$, coordinate domination follows from
\[
 \Phi(\bm z)-z_b=P(s_b)-s_b-c_\Phi+\sum_{c\ne b}P(s_c)\ge0.
\]
Constant row shifts leave $\bm q$ unchanged and shift $\Phi$ by the same constant, so $\sqrt{\quadform}$ is a seminorm.

For comparison, \cref{lem:scalar-bounds} and $|\widetilde w_b|\le2\norm{\bm w}_\infty$ imply
\[
 |\nabla_{\bm w}\ln h_b|=|h'(q_b)\widetilde w_b|\le4\norm{\bm w}_\infty,
 \qquad
 |\nabla_{\bm w}\ln\alpha_b|\le\delta\norm{\bm w}_\infty.
\]
Integrating along the segment proves \eqref{eq:local-comparison-full}. Compare the weighted sums for every fixed centering constant and then minimize to obtain \eqref{eq:quadratic-comparison-full}.

It remains to bound $\Phi(0)$. Let $s_{m_i}=f^{-1}(1/m_i)$ and $s_1=f^{-1}(1)$. Changing variables and integrating by parts gives
\begin{equation*}
 \Phi(0)=(m_i-1)\int_0^{1/m_i}\frac q{h(q)}\,dq
 +\int_{1/m_i}^1\frac{1-q}{h(q)}\,dq.
\end{equation*}
Set $q_0=f(0)\le e\beta/2<1/m_i$. Since $|ds/du|=1+u^{-2}$ on $u\ge1$,
\[
 P(0)\le2e\beta\int_1^\infty\frac{du}{u(1+\delta u/4)}
 =2e\beta\ln(1+4/\delta).
\]
For $q\in[q_0,1]$, the positive-region argument gives $h(q)\ge(2/5)q$. Splitting the first integral at $q_0$ and integrating the remaining elementary terms yields
\[
 \Phi(0)\le m_iP(0)+\frac52\ln m_i.
\]
For the public parameters, $m_i\beta=(\log m_i)/(2^{10}J)$ and $J\ge2$. Since $1+4\cdot2^J\le5^J$,
$\ln(1+4/\delta)\le J\ln5$, and therefore
\[
 0\le\Phi(0)
 \le\left(\frac52\ln2+\frac{2e\ln5}{2^{10}}\right)\log m_i<2\log m_i.
\]
\end{proof}

\section{Uniform bounds for normalized derivatives}\label{app:normalized-derivatives}
This appendix proves \cref{lem:normalized-derivatives-full},
which supplies the local derivative bounds used in both prediction
stages. The ordinary bound controls error-direction derivatives, and
the order-raising identity represents gain-direction derivatives in
the early stage. In the late stage, the curvature-weighted bound
controls derivatives of existing atoms and residual factors directly:
their score dependence is still differentiated, but their orders no
longer increase.

For $e=(i,a,b)$ and $\bm v\in[0,1]^{m_i}$, define
\begin{equation*}
 \chi^\lambda_{e,\bm v}(\bm y)
 =\frac{h_e(\bm y)}{q_e(\bm y)}
 \bigl(v_b-\ip{\bm\sigma_{i,a}(\bm y)}{\bm v}\bigr),
 \qquad
 \chi^\alpha_{e,\bm v}(\bm y)
 =\alpha_e(\bm y)
 \bigl(v_b-\ip{\bm\sigma_{i,a}(\bm y)}{\bm v}\bigr).
\end{equation*}
They depend only on the selected row and are invariant under a constant shift of that row.

\begin{lemma}[Uniform normalized derivatives]\label{lem:normalized-derivatives-full}
There is a universal integer $B\ge1$. Set $B_r=B^{r+1}r!$, and for $\bm v_0,\ldots,\bm v_r\in[0,1]^{m_i}$ in row $(i,a)$ define
\begin{equation*}
 \mathcal J^{(r)}_{e,\bm v_0,\ldots,\bm v_r}(\bm y)
 =B_r^{-1}\nabla_{\bm v_r}\cdots\nabla_{\bm v_1}\chi^\lambda_{e,\bm v_0}(\bm y).
\end{equation*}
Uniformly in $r\ge0$,
\begin{equation}\label{eq:J-bounds-full}
 |\mathcal J^{(r)}|\le1,
 \qquad
 |\nabla_{\bm w}\mathcal J^{(r)}|
 \le C\min\left\{\norm{\bm w_{i,a}}_\infty,
 \beta_i^{-1/2}\sqrt{\quadform_{\bm y_{i,a}}(\bm w_{i,a})}\right\},
\end{equation}
\begin{equation*}
 \nabla_{\bm v_{r+1}}\mathcal J^{(r)}=B(r+1)\mathcal J^{(r+1)}.
\end{equation*}
Each factor is multilinear in its payoff directions and row-shift invariant. Also $|\chi^\alpha|\le\alpha_e\le1$, and its first derivative obeys the same minimum bound as in \eqref{eq:J-bounds-full}.
\end{lemma}

We first prove the $\chi^\alpha$ statement. Suppress row indices and put $h_{\mathrm{total}}=\sum_bh_b$. By \cref{lem:scalar-bounds}, $h_{\mathrm{total}}\ge3/10$. Differentiating $\sigma_b=h_b/h_{\mathrm{total}}$ gives
\begin{equation*}
 \nabla_{\bm w}\ip{\bm\sigma}{\bm v}
 =\Cov_{b\sim\bm\sigma}(v_b,h'(q_b)\widetilde w_b),
 \qquad
 \left|\nabla_{\bm w}\ip{\bm\sigma}{\bm v}\right|
 \le\sqrt{\quadform_{\bm y}(\bm w)/h_{\mathrm{total}}}\le C\sqrt{\quadform_{\bm y}(\bm w)}.
\end{equation*}
Moreover, \cref{lem:scalar-bounds} gives
\begin{equation*}
 |\nabla_{\bm w}\alpha_e|
 \le\frac\delta2\alpha_e^2|\widetilde w_e|
 \le\frac{\delta\alpha_e^2}{2\sqrt{h_e}}\sqrt{\quadform_{\bm y}(\bm w)}
 \le C\beta_i^{-1/2}\sqrt{\quadform_{\bm y}(\bm w)}.
\end{equation*}
The product rule proves the curvature-weighted bound. Using $|\widetilde w_b|\le2\norm{\bm w}_\infty$ instead proves the ordinary bound.

It remains to prove the $\chi^\lambda$ statement. The argument has four steps.

\subsection{Scalar complex-analytic bounds}
Fix real $s_0$ and write $u_0=u(s_0)>0$. For $|\zeta|\le1/8$,
\[
 |(s_0+\zeta)^2+4|
 \ge s_0^2+4-2|s_0||\zeta|-|\zeta|^2
 \ge\frac34(s_0^2+4).
\]
Choose the square-root branch positive at $s_0$. The corresponding $u$ is nonzero, because $(\sqrt{s^2+4}-s)(\sqrt{s^2+4}+s)=4$. The local logarithm satisfies $(\ln u)'=-1/\sqrt{s^2+4}$, hence
\begin{equation}\label{eq:complex-u}
 \left|\ln\frac{u(s_0+\zeta)}{u_0}\right|\le|\zeta|.
\end{equation}
For every $c\ge0$, the relative error of $1+cu$ from $1+cu_0$ is at most $e^{|\zeta|}-1$. For $1+u^2$ it is at most $e^{2|\zeta|}-1$. The three terms in \eqref{eq:lambda-u} are therefore holomorphic, nonzero, and have uniformly bounded logarithmic derivatives on a smaller universal disk. Their real base values are positive, so
\begin{equation*}
 |\lambda(s_0+\zeta)-\lambda(s_0)|\le C|\zeta|\lambda(s_0),
 \qquad
 |\lambda(s_0+\zeta)|\le2\lambda(s_0)\le2.
\end{equation*}
Integrating $f'/f=\lambda$ and using $f'=f\lambda$, on a sufficiently small universal radius $r_s>0$,
\begin{equation}\label{eq:complex-scalar-bounds}
 |f(s_0+\zeta)-f(s_0)|\le C|\zeta|f(s_0),
 \qquad
 |f'(s_0+\zeta)-f'(s_0)|\le C|\zeta|f'(s_0),
\end{equation}
\begin{equation*}
 \left|\frac{f''(s_0+\zeta)}{f'(s_0+\zeta)}\right|\le C,
 \qquad
 |\lambda'(s_0+\zeta)|\le C\sqrt{f'(s_0)/\beta}.
\end{equation*}
No factor $d$ appears: it cancels in the logarithmic derivative. For the last bound, if $u_0\le1$, use \eqref{eq:positive-h-lower} and the uniform bound on $\lambda'$. If $u_0\ge1$, the first term of \eqref{eq:lambda-u} is bounded by $Cu_0^{-2}$, the other terms by $Cu_0^{-1}$, and their logarithmic derivatives by $Cu_0^{-1}$ uniformly by \eqref{eq:complex-u}, after enlarging the universal constant $C$. Thus $|\lambda'|\le C/u_0^2$. Equation \eqref{eq:negative-comparisons} gives
\[
 \sqrt{f'(s_0)/\beta}\ge\frac{c}{u_0\sqrt{1+\delta u_0}},
\]
which dominates $u_0^{-2}$ because $\sqrt{1+\delta u_0}\le\sqrt{2u_0}$.

\subsection{A dimension-uniform normalization radius}
Fix a real normalized row, write $s_b=z_b-\mu_0$, $h_b^0=f'(s_b)$, and $h_{\mathrm{total}}^0=\sum_bh_b^0\ge3/10$. For a complex score displacement $\bm\zeta$, define
\begin{equation*}
 F(\bm\zeta,\nu)=\sum_bf(s_b+\zeta_b-\nu)-1,
 \qquad
 \Psi_{\bm\zeta}(\nu)=\nu+F(\bm\zeta,\nu)/h_{\mathrm{total}}^0.
\end{equation*}
Choose $r_0>0$ so that $3r_0<r_s$ and the relative $f'$ errors in \eqref{eq:complex-scalar-bounds} are at most $1/10$ on radius $3r_0$. For $\norm{\bm\zeta}_\infty<r_0$ and $|\nu|\le2r_0$,
\[
 |\partial_\nu\Psi_{\bm\zeta}|
 =\left|1-\frac{\sum_bf'(s_b+\zeta_b-\nu)}{h_{\mathrm{total}}^0}\right|\le\frac1{10},
 \qquad
 |\Psi_{\bm\zeta}(0)|\le\frac{11}{10}r_0.
\]
Thus $\Psi_{\bm\zeta}$ is a contraction of the closed disk $|\nu|\le2r_0$ into itself. Iteration from zero converges uniformly, and its iterates are holomorphic. The limit $\nu(\bm\zeta)$ is holomorphic and satisfies $F(\bm\zeta,\nu(\bm\zeta))=0$. It extends the real normalizer by $\mu(\bm z+\bm\zeta)=\mu_0+\nu(\bm\zeta)$.

At these complex scores, let $h_b=f'(s_b+\zeta_b-\nu)$, $h_{\mathrm{total}}=\sum_bh_b$, and $\sigma_b=h_b/h_{\mathrm{total}}$. Relative bounds give
\begin{equation}\label{eq:complex-row-bounds}
 |h_b|\le2h_b^0,
 \qquad
 |h_{\mathrm{total}}|\ge h_{\mathrm{total}}^0/2,
 \qquad
 \sum_b|\sigma_b|\le4.
\end{equation}
The radius is independent of the row dimension and of the smallest probability.

\subsection{Retaining one direction at real curvature}
For fixed real $\bm v\in[0,1]^{m_i}$, write
\[
 \chi(\bm z+\bm\zeta)
 =\lambda(s_b+\zeta_b-\nu)
 \bigl(v_b-\ip{\bm\sigma}{\bm v}\bigr).
\]
We claim that every real direction $\bm w$ satisfies
\begin{equation}\label{eq:one-weighted-derivative}
 \sup_{\norm{\bm\zeta}_\infty<r_0}|\nabla_{\bm w}\chi(\bm z+\bm\zeta)|
 \le C\min\left\{\norm{\bm w}_\infty,
 \beta^{-1/2}\sqrt{\quadform_{\bm z}(\bm w)}\right\}.
\end{equation}
Complex normalization and $\chi$ are invariant under a common score shift. For the weighted bound, subtract the real mean $\sum_bh_b^0w_b/h_{\mathrm{total}}^0$, so $\sum_bh_b^0w_b=0$, and set $Q_0=\sum_bh_b^0w_b^2=\quadform_{\bm z}(\bm w)$. Differentiating normalization gives $\nabla_{\bm w}\mu=\sum_bh_bw_b/h_{\mathrm{total}}$. By \eqref{eq:complex-row-bounds} and Cauchy--Schwarz, for $\widetilde w_b=w_b-\nabla_{\bm w}\mu$,
\begin{equation*}
 |\nabla_{\bm w}\mu|\le C\sqrt{Q_0/h_{\mathrm{total}}^0},
 \qquad
 |\widetilde w_b|\le C\sqrt{Q_0/h_b^0},
 \qquad
 \sum_bh_b^0|\widetilde w_b|^2\le CQ_0.
\end{equation*}
Writing $\varkappa_b=f''/f'$ at the complex scalar arguments,
\[
 \nabla_{\bm w}\ip{\bm\sigma}{\bm v}
 =\frac1{h_{\mathrm{total}}}\sum_bh_b\varkappa_b\widetilde w_b
 \bigl(v_b-\ip{\bm\sigma}{\bm v}\bigr),
 \qquad
 \left|\nabla_{\bm w}\ip{\bm\sigma}{\bm v}\right|\le C\sqrt{Q_0}.
\]
The last bound uses $|\varkappa_b|\le C$, \eqref{eq:complex-row-bounds}, and $h_{\mathrm{total}}^0\ge3/10$. The product rule gives
\[
 \nabla_{\bm w}\chi
 =\lambda'\widetilde w_b\bigl(v_b-\ip{\bm\sigma}{\bm v}\bigr)
 -\lambda \nabla_{\bm w}\ip{\bm\sigma}{\bm v}.
\]
The first term is bounded by
$\sqrt{h_b^0/\beta}\sqrt{Q_0/h_b^0}=\beta^{-1/2}\sqrt{Q_0}$, and the second by $C\sqrt{Q_0}$. For the ordinary bound, do not center $\bm w$ at the real point. The same formulas give $|\nabla_{\bm w}\mu|\le4\norm{\bm w}_\infty$ and $|\widetilde w_b|\le5\norm{\bm w}_\infty$. This proves \eqref{eq:one-weighted-derivative}. The same bounds also give $|\chi|\le10$.

\subsection{Cauchy normalization}
For $r\ge1$ and directions of sup norm at most one, the map
\[
 (t_1,\ldots,t_r)\longmapsto
 \chi\!\left(\bm z+\sum_jt_j\bm v_j\right)
\]
is holomorphic on $|t_j|<r_0/(2r)$. Repeated Cauchy bounds give
\begin{equation*}
 |\nabla_{\bm v_r}\cdots\nabla_{\bm v_1}\chi|
 \le10(2r/r_0)^r\le10(2e/r_0)^rr!.
\end{equation*}
Apply the same bounds to the holomorphic function $\nabla_{\bm w}\chi(\bm z+\bm\zeta)$ using \eqref{eq:one-weighted-derivative}. Only $r$ further derivatives are taken, so
\begin{equation*}
 |\nabla_{\bm v_r}\cdots\nabla_{\bm v_1}\nabla_{\bm w}\chi|
 \le C(2e/r_0)^rr!
 \min\left\{\norm{\bm w}_\infty,
 \beta^{-1/2}\sqrt{\quadform_{\bm z}(\bm w)}\right\}.
\end{equation*}
Choose a universal integer $B$ at least $2e/r_0$, $10$, and the finite constants in these bounds. Divide by $B_r=B^{r+1}r!$. This proves \cref{lem:normalized-derivatives-full}, including $r=0$, multilinearity, row-shift invariance, and the recursion $B_{r+1}/B_r=B(r+1)$. The weighted quadratic form remains at the fixed real row throughout. No complex curvature comparison or minimum-probability bound is used.

\section{Stationary trees, centered gain identities, and moment bounds}\label{app:trees-moments}
This appendix supplies the probabilistic inputs to both prediction stages. The first subsection represents each deviation gain under the stationary-tree distribution. The second proves an unweighted player-level tree-score bound for the early stage and a curvature-weighted tree-score bound for the late stage. The third proves the moment bound that permits products of overlapping centered atoms.

\subsection{The Markov chain tree theorem}
Let $Q$ be a strictly positive row-stochastic matrix. For a root $a$, set
\[
 w_a=\sum_{\tau\text{ directed toward }a}\prod_{(c,d)\in\tau}q_{c,d}.
\]
Every $w_a$ is positive. For a fixed $b$, take a tree rooted at $a\ne b$, add the edge $a\to b$, and delete the outgoing edge $b\to c$ on the resulting cycle. This gives a weight-preserving bijection with a tree rooted at $b$ together with a distinguished edge $b\to c$. Hence
\[
 \sum_{a\ne b}w_aq_{a,b}=w_b\sum_{c\ne b}q_{b,c}.
\]
Therefore $x_a=w_a/\sum_cw_c$ satisfies $\bm x^\top(Q-I)=0$. If $(Q-I)\bm v=0$ and $v_a$ is maximal, then
\[
 0=\sum_{b\ne a}q_{a,b}(v_b-v_a).
\]
Every summand is nonpositive and every coefficient is positive, so every $v_b=v_a$. The right kernel, and hence the left kernel, is one-dimensional. Thus the normalized vector above is the unique stationary distribution. Replacing one redundant stationarity equation by $\sum_ax_a=1$ gives the nonsingular linear system used by the exact stationary-distribution solve in the learning dynamics.

For each player and score array $\bm y$, let $\treeP_{i,\bm y}$ be the corresponding normalized distribution on rooted directed trees, and let $\treeP_{\bm y}=\bigotimes_i\treeP_{i,\bm y}$. Write $I_e$ for edge inclusion, with diagonal indicators zero, and $a_i(\bm T)$ for the root of player $i$'s tree. A tree has $m_i-1$ edges and at most one outgoing edge from each source row.

Define
\begin{equation*}
 \bm v_i(\bm T)=\bigl(u_i(b,\bm a_{-i}(\bm T))\bigr)_{b\in[m_i]},
 \qquad
 G_{i,a,b}(\bm T)=\one_{\{a_i(\bm T)=a\}}
 \bigl(u_i(b,\bm a_{-i}(\bm T))-u_i(a,\bm a_{-i}(\bm T))\bigr).
\end{equation*}
Then $\sourcenorm{\bm G_i(\bm T)}\le1$ and
\begin{equation}\label{eq:gain-tree-representation}
 \bm r_i^{(s)}=\E_{\bm T\sim\treeP_{\bm z^{(s)}}}\bm G_i(\bm T).
\end{equation}
At trial scores $\bm y$, define the centered represented gain
\begin{equation*}
 \widetilde r_{i,a,b}^{(s)}(\bm y)
 =x_{i,a}^{(s)}\left(v_{i,b}^{(s)}-\ip{\bm\sigma_{i,a}(\bm y)}{\bm v_i^{(s)}}\right).
\end{equation*}
For either $\zeta=\alpha$ or $\zeta=h/q$,
\begin{equation}\label{eq:centered-tree-identity}
 \zeta_e(\bm y)\widetilde r_e^{(s)}(\bm y)
 =\E_{\bm T\sim\treeP_{\bm z^{(s)}}}
 \left[\one_{\{a_i(\bm T)=a\}}\chi^\zeta_{e,\bm v_i(\bm T)}(\bm y)\right],
\end{equation}
where $\chi^{h/q}=\chi^\lambda$. The identity follows from independence of player $i$'s root and the opponents' roots. A fresh independent sample makes it valid inside an expectation involving old samples.

\subsection{Tree score bounds}
\begin{lemma}[Tree score bounds]\label{lem:tree-score-full}
For a deterministic score direction $\bm w$,
\begin{equation}\label{eq:tree-score-full}
 \nabla_{\bm w}\ln\treeP_{\bm y}
 =\sum_e(I_e-\E I_e)\frac{h_e}{q_e}\widetilde w_e
 =\sum_e(I_e-\E I_e)(\gamma_e+\delta\alpha_e)\widetilde w_e.
\end{equation}
Moreover,
\begin{equation}\label{eq:tree-source-variance}
 \norm{\nabla_{\bm w}\ln\treeP_{\bm y}}_{L^2}
 \le2\left(\sum_i\sourcenorm{\bm w_i}^2\right)^{1/2},
\end{equation}
\begin{equation}\label{eq:tree-gamma-variance}
 \E\left|\sum_e(I_e-\E I_e)\gamma_e\widetilde w_e\right|^2
 \le\sum_{i,a}\omega_i\quadform_{\bm y_{i,a}}(\bm w_{i,a}).
\end{equation}
If a normed-vector integrand $\bm F$ satisfies $\E\norm{\bm F}^2\le1$, these inequalities also bound the norm of the expectation of $\bm F$ times the corresponding centered score.
\end{lemma}

\begin{proof}
A tree's log-weight derivative is $\sum_eI_e(h_e/q_e)\widetilde w_e$, by normalized row differentiation. Differentiating the tree normalizer subtracts its mean, proving \eqref{eq:tree-score-full}. The row center is essential even before the tree mean is subtracted.

Since $h_e/q_e\le1$, player $i$'s uncentered tree score has absolute value at most $2\sourcenorm{\bm w_i}$. Independence across players proves \eqref{eq:tree-source-variance}. For deterministic coefficients $c_e$, Cauchy--Schwarz on the $m_i-1$ tree edges gives
\begin{equation*}
 \Var\left(\sum_eI_ec_e\right)
 \le\sum_i(m_i-1)\sum_{e\in\cI_i}c_e^2.
\end{equation*}
Take $c_e=\gamma_e\widetilde w_e$, use $\gamma_e^2\le C_\gamma h_e/\beta_i$, and recall $\omega_i=C_\omega(m_i-1)/\beta_i$ with $C_\omega\ge C_\gamma$. This proves \eqref{eq:tree-gamma-variance}. Finally,
\[
 \norm{\E(\bm F X)}\le(\E\norm{\bm F}^2)^{1/2}(\E|X|^2)^{1/2}.
\]
Here $\bm F$ need not be independent of the centered score.
\end{proof}
Bound \eqref{eq:tree-source-variance} is used in the early closure, while \eqref{eq:tree-gamma-variance} is used in the late closure.

\subsection{A moment bound for overlapping centered factors}
\begin{lemma}[Centered sums and products]\label{lem:centered-moment}
Conditional on an auxiliary object $D_0$, suppose the pairs $(\bm T_i,\bm T_i')$ are independent across players and iid within each pair. Let $F_i(\cdot,D_0)\in[-1,1]$, and let fixed real weights satisfy $\sum_ia_i^2\le1$. For integer $\ell\ge1$, define
\begin{equation*}
 Z=\frac1{2\sqrt\ell}\sum_ia_i
 \bigl(F_i(\bm T_i,D_0)-F_i(\bm T_i',D_0)\bigr).
\end{equation*}
Then $\norm Z_{L^{2\ell}}\le1$. If at most $\ell$ factors each have this marginal bound, their product has $L^2$ norm at most one, regardless of dependencies between factors.
\end{lemma}

\begin{proof}
Conditional Hoeffding bounds give
\[
 \E\exp\left(t\sum_ia_i(F_i-F_i')\right)\le e^{t^2}.
\]
The resulting two-sided tail $2e^{-u^2/4}$ implies
\[
 \E\left|\sum_ia_i(F_i-F_i')\right|^{2\ell}
 \le2\,4^\ell\ell!.
\]
For $\ell\ge2$, $2\ell!\le\ell^\ell$. For $\ell=1$, variance at most two suffices. Remove the conditioning. For $h\le\ell$, H\"older gives
\[
 \E\prod_{j=1}^h|Z_j|^2
 \le\prod_{j=1}^h(\E|Z_j|^{2h})^{1/h}\le1.
\]
\end{proof}
This moment bound supplies the $L^2$ product bound used in both expression classes.

\section{The early player-centered closure}\label{app:early-closure}
This appendix proves \cref{lem:early-closure} and the handoff bound \eqref{eq:early-mass-filtered}. Throughout, we fix the finite deterministic effective-score history. Integrals may include finite sums and compact parameter intervals.

\subsection{Score types, slots, and atoms}

A \emph{score type} is a sequence
\begin{equation}\label{eq:score-type}
 \bm y_\tau^{(t)}=\sum_{a=0}^pc_{\tau,a}\bm z^{(t-a)},
 \qquad
 c_{\tau,a}\ge0,
 \qquad
 \sum_ac_{\tau,a}=1,
\end{equation}
with coefficients independent of time and common to all players. A \emph{sample slot} $q$ is an independent product-tree draw
$\bm T^q\sim\treeP_{\bm y_{\tau(q)}^{(t)}}$. Different slots may have the same type without being the same sample. Write $\Omega=(\bm T^1,\ldots,\bm T^M)$ for their tuple and $\E_t$ for expectation under their time-$t$ product distribution.

An order-$r$ atom uses two distinct inclusion slots $q,q'$ of one type $\tau$, and pairwise distinct direction slots $d_0,\ldots,d_r$ disjoint from $q,q'$. Put
\[
 a_i^0=a_i(\bm T^{d_0}),
 \qquad
 \bm v_i^j=\bm v_i(\bm T^{d_j}),
\]
and
\begin{equation*}
 F_i(\bm T^q)
 =\left(\prod_{j=1}^r\one_{\{a_i(\bm T^{d_j})=a_i^0\}}\right)
 \sum_b I_{i,a_i^0,b}(\bm T_i^q)
 \mathcal J^{(r)}_{i,a_i^0,b,\bm v_i^0,\ldots,\bm v_i^r}(\bm y_\tau).
\end{equation*}
For fixed weights $\sum_ia_i^2\le1$, define
\begin{equation*}
 Z^{(r)}(\Omega,\bm y_\tau^{(t)})=\frac1{2\sqrt\ell}\sum_ia_i
 \bigl(F_i(\bm T^q)-F_i(\bm T^{q'})\bigr).
\end{equation*}
We suppress the sample and score arguments when unambiguous. Each selected row has at most one outgoing tree edge, so $|F_i|\le1$. Conditioning on this atom's direction slots leaves its inclusion pair iid and independent across players. Hence \cref{lem:centered-moment} gives $\norm{Z^{(r)}}_{L^{2\ell}}\le1$. Different atoms may share slots, even in different roles. We apply the marginal bound separately and use unconditional H\"older for products.

A depth-$p$ early expression, $0\le p\le\ell$, has the form
\begin{equation*}
 \bm U^{(t)}=\one_{\{t>p\}}\E_t\left[\bm G(\Omega)
 \prod_{\nu=1}^h Z_\nu^{(r_\nu)}(\Omega,\bm y_{\tau_\nu}^{(t)})\right],
 \qquad
 \sum_{\nu=1}^h(r_\nu+1)\le p.
\end{equation*}
There are at most $1+2p$ independent slots, and all types have the window \eqref{eq:score-type}. The observable is a fixed matrix-valued function of the sample tuple with $\sourcenorm{\bm G(\Omega)}\le1$. Initially it is $\bm G_i(\bm T^1)$ from \eqref{eq:gain-tree-representation}. The observable, slot assignments, atom weights, and type coefficients are independent of time. Holding the samples fixed, all score dependence lies in the normalized factors. The expectation also depends on the scores through the sampling distributions. By \cref{lem:centered-moment}, the integrand has squared norm expectation at most one, so
\begin{equation*}
 \sourcenorm{\bm U^{(t)}}\le1.
\end{equation*}
The gain representation \eqref{eq:gain-tree-representation} is a depth-zero expression. We call an index $s$ attached to a gain direction $\bm r^{(s)}$ introduced by this representation a \emph{represented gain time}. The leading indicator $\one_{\{t>p\}}$ is the \emph{gate}. It ensures that every represented gain time introduced later is positive.

\subsection{One-difference closure}

\begin{proof}[Proof of \cref{lem:early-closure}]
The gate contributes one impulse at $t=p+1$, of source norm at most one. For $t\ge p+2$, interpolate every score type simultaneously between its value at $t-1$ and $t$. For independent slots with distributions $P_q(\theta)$, finite differentiation is exactly
\begin{equation}\label{eq:differentiate-product-distributions}
 \frac d{d\theta}\E_\theta F_\theta
 =\E_\theta\frac{dF_\theta}{d\theta}
 +\sum_q\E_\theta\left[F_\theta\frac d{d\theta}\ln P_q(\theta)\right].
\end{equation}
Equal types move together, but each independent slot density contributes separately. By \eqref{eq:effective-score-difference}, every type direction is a convex combination of
$\eta_i\bm r_i^{(s)}-\eta_i\diff\bm\xi_i^{(s)}$ at positive represented times. Treat one shift at a time.

\paragraph{Error directions stop.}
For a distribution derivative, use \eqref{eq:tree-source-variance} with $\bm w_i=-\eta_i\diff\bm\xi_i$. Since backward differences have kernel mass two,
\[
 \sum_i\eta_i^2\sum_{t=1}^T\sourcenorm{(\diff\bm\xi_i)^{(t)}}^2
 \le4g^2E^2.
\]
The sequence cost is therefore at most $CgE$.

For an atom, freeze its samples. At most two outgoing edges in its inclusion pair contribute per player. By \eqref{eq:J-bounds-full} and Cauchy--Schwarz over the deterministic atom weights,
\begin{equation}\label{eq:atom-ordinary-derivative}
 |\nabla_{\bm w}Z^{(r)}|
 \le\frac C{\sqrt\ell}
 \left(\sum_i\sourcenorm{\bm w_i}^2\right)^{1/2}.
\end{equation}
After removing the differentiated atom, the remaining product has $L^1$ norm at most one. Its error contribution is again at most $CgE$. There are at most $1+2p$ distribution densities and at most $p$ atoms, so, after enlarging a universal constant, all stopped error terms cost at most $C_E(p+1)gE$.

\paragraph{Gain-direction distribution derivatives continue.}
When a slot $q$ is differentiated, introduce an independent ghost $q'$ of the same type. Inside the expectation, replace $I_e(\bm T^q)-\E I_e$ by $I_e(\bm T^q)-I_e(\bm T^{q'})$. The ghost is independent of the complete old integrand. Introduce a separate fresh gain slot $d_0$ at the represented time $s$, using \eqref{eq:centered-tree-identity} with $\zeta=h/q$. The multiplier is
\begin{equation*}
 \sum_i\eta_i\sum_{a,b}
 \bigl[I_{i,a,b}(\bm T_i^q)-I_{i,a,b}(\bm T_i^{q'})\bigr]
 \one_{\{a_i(\bm T^{d_0})=a\}}
 \chi^\lambda_{i,a,b,\bm v_i(\bm T^{d_0})}(\bm y_\tau).
\end{equation*}
Since $\chi^\lambda=B\mathcal J^{(0)}$, this is $2Bg\sqrt\ell$ times a new order-zero atom with weights $a_i=\eta_i/g$. It adds two slots and one degree unit.

\paragraph{Gain-direction atom derivatives raise the order.}
Row-shift invariance removes the constant source payoff. For an order-$r$ normalized factor,
\begin{equation*}
 \nabla_{\eta_i\bm r_{i,a}^{(s)}}\mathcal J^{(r)}
 =\eta_ix_{i,a}^{(s)}B(r+1)
 \mathcal J^{(r+1)}_{\ldots,\bm v_i^{(s)}}.
\end{equation*}
A fresh gain slot $d_{r+1}$ represents both the source probability and the additional payoff argument, by multilinearity and root independence. Factor out $\eta_{\max}B(r+1)$ and replace the atom weights by $a_i\eta_i/\eta_{\max}$. Their squared sum remains at most one. One slot and one degree unit are added.

These are all product-rule terms. Distribution-density continuations have total coefficient mass at most $C(1+2p)g\sqrt\ell$. Explicit-factor continuations have mass at most
$Bg\sum_\nu(r_\nu+1)\le Bgp$. This proves the stated coefficient bound after fixing $C_E$. Old types interpolate into the enlarged window, ghosts inherit a type, and new gain slots have a single-time type. Every coefficient remains independent of time, and the children satisfy all slot, disjointness, and degree conditions.
\end{proof}

\paragraph{Early coefficient mass.}

Starting from a depth-zero gain, let $a_p$ bound the total absolute coefficient mass at depth $p$. The closure gives
\begin{equation*}
 a_p\le(C_Eg\sqrt\ell)^pp!.
\end{equation*}
For $p\le\ell$, $p!\le p^p\le k^p$, and therefore
\begin{equation}\label{eq:early-mass-filtered}
 a_pk^p\le(C_Egk^{5/2})^p=\kappa^p.
\end{equation}
\section{The late transition-index closure}\label{app:late-prediction}
This appendix proves \cref{lem:late-closure}. It records the distinct transition indices encountered along each branch: curvature-controlled terms and repeated indices stop, while only a fresh index continues.

\subsection{Late expressions and their score window}

A depth-$j$ late expression, $0\le j\le N$, has the form
\begin{equation*}
 \bm U^{(t)}=\one_{\{t>\ell+j\}}\E_t\left[
 \bm G(\Omega)
 \prod_{\nu=1}^hZ_\nu^{(r_\nu)}(\Omega,\bm y_{\tau_\nu}^{(t)})
 \prod_{\mu=1}^j
 \chi^\alpha_{e_\mu(\Omega),\bm v_\mu(\Omega)}
 (\bm y_{\upsilon_\mu}^{(t)})
 \right].
\end{equation*}
There are at most $1+2\ell+2j$ independent slots. Score types are fixed convex combinations in the depth-$(\ell+j)$ window. All early-atom conditions hold with total degree at most $\ell$. The observable, transition indices, and payoff directions depend only
on samples and fixed parameters, with $\sourcenorm{\bm G}\le1$.
Whenever $e_\mu(\Omega)=(i,a,b)$, the associated payoff vector satisfies
$\bm v_\mu(\Omega)\in[0,1]^{m_i}$. The $j$ transition indices are distinct wherever $\bm G\ne0$, and every scalar residual factor has a fixed score type even when its transition index is sample-dependent. By \cref{lem:centered-moment} and $|\chi^\alpha|\le1$, the integrand has squared norm expectation at most one. An early depth-$\ell$ expression is a late depth-zero expression.

We first record the consequence of the energy normalization used throughout this appendix. If $s\ge1$ and $\bm y$ lies within distance $\epsilon$ of $\bm z_{i,a}^{(s)}$, then \eqref{eq:quadratic-comparison-full} gives
\begin{equation*}
 \eta_i\sqrt{\quadform_{\bm y}(\bm r_{i,a}^{(s)})}
 \le\sqrt8\,e^{2\epsilon}\sqrt{b_{i,a}^{(s)}}.
\end{equation*}
Keeping one summand in the quadratic form gives the same bound for $\eta_i\sqrt{h_e(\bm y)}\,|\widetilde r_e^{(s)}(\bm y)|$.

\begin{lemma}[Comparison on the late score window]\label{lem:late-window}
In a depth-$j$ difference, let $\bm y,\bm y'$ lie in the convex hull of
$\bm z^{(t)},\ldots,\bm z^{(t-\ell-j-1)}$, and let $s\ge1$ be a represented gain time in this window. For $e=(i,a,b)$,
\begin{equation}\label{eq:late-curvature-comparison}
 \eta_i\sqrt{\quadform_{\bm y_{i,a}}(\bm r_{i,a}^{(s)})}
 \le C\,2^j\sqrt{b_{i,a}^{(s)}},
\end{equation}
\begin{equation}\label{eq:late-repetition-comparison}
 \delta\eta_i\alpha_e(\bm y)\alpha_e(\bm y')
 |\widetilde r_e^{(s)}(\bm y)|
 \le C\,2^j\left(\sum_{i',a'}\omega_{i'}b_{i',a'}^{(s)}\right)^{1/2}.
\end{equation}
\end{lemma}

\begin{proof}
From $\bm z=\bm\theta-\eta\bm\xi$, the gain bound and \eqref{eq:pointwise-bounds}, the player-$i$ hull diameter is at most $\eta_i(\ell+j+195)$. The comparisons in \eqref{eq:local-comparison-full}--\eqref{eq:quadratic-comparison-full} therefore cost at most
\[
 C\exp\bigl((2+\delta)\eta_{\max}(\ell+j+195)\bigr)\le C\,2^j,
\]
where the last inequality uses $c\le1/16$, $k\ge5$, and $\delta\le1/2$. This proves \eqref{eq:late-curvature-comparison}. For \eqref{eq:late-repetition-comparison}, compare the two $\alpha$ values, use
$\delta\alpha_e^2\le C\sqrt{h_e/\beta_i}$, retain the corresponding summand in $\quadform$, and use $\omega_i\ge1/\beta_i$. This is a deterministic, index-uniform bound before any new gain sample is introduced.
\end{proof}

\subsection{One-difference closure}

\begin{proof}[Proof of \cref{lem:late-closure}]
The gate contributes one impulse. At later times, interpolate equal score types together and apply \eqref{eq:differentiate-product-distributions}. Directions are convex combinations of $\eta_i\bm r_i^{(s)}-\eta_i\diff\bm\xi_i^{(s)}$ at positive represented times.

\paragraph{Error directions stop.}
Use \cref{lem:tree-score-full} for a distribution, \eqref{eq:atom-ordinary-derivative} for an early atom, and the ordinary derivative bound in \cref{lem:normalized-derivatives-full} for a scalar residual factor. Each has sequence cost at most $CgE$, even when a selected transition index is sample-dependent.

\paragraph{Curvature-controlled gain derivatives stop.}
For a distribution, split its centered score as in \eqref{eq:tree-score-full}. Keep the entire $\gamma$ score intact and use \eqref{eq:tree-gamma-variance} with the complete integrand. Its $L^2$ bound and \cref{lem:late-window} give cost $C\,2^jY$.

For an existing atom, freeze the samples and use the curvature-weighted branch of \eqref{eq:J-bounds-full}, as in the proof of \eqref{eq:atom-ordinary-derivative}, to obtain
\begin{equation*}
 |\nabla_{\eta\bm r^{(s)}}Z^{(r)}|
 \le\frac C{\sqrt\ell}
 \left(\sum_{i,a}\beta_i^{-1}\eta_i^2
 \quadform_{\bm y_{i,a}}(\bm r_{i,a}^{(s)})\right)^{1/2}.
\end{equation*}
Remove that atom and bound the remaining product in $L^1$. Since $\omega_i\ge1/\beta_i$, \cref{lem:late-window} gives cost $C\,2^jY$. Existing scalar residual factors stop by the same weighted derivative bound. Thus no old normalized derivative order increases in the late stage.

\paragraph{Repeated transition indices stop.}
If the $\delta\alpha$ part of a distribution derivative selects an already present transition index $e$, its old factor supplies $\alpha_e(\bm y')$, while its centered payoff has magnitude at most one. Apply \eqref{eq:late-repetition-comparison} while the gain direction is still deterministic. The centered edge indicator has magnitude at most one. Removing the old factor that pays for the repetition leaves an $L^1$-bounded product. At most $j$ transition indices can repeat on any sample outcome, so this costs at most $Cj2^jY$ per distribution.

There are at most $1+3\ell+3j$ distribution densities and explicit factors. Summing these payments and the gate gives
\begin{align*}
 \sourcenormT{\bm V}
 &\le1+C(\ell+j+1)gE+C(\ell+j+1)^2 2^jY\\
 &\le C_Lk^2 8^j(1+gE+Y).
\end{align*}
For every fixed time shift, the squared curvature sequence sums to at most $Y^2$. Convex shift weights and interpolation do not enlarge this finite-prefix bound. The last inequality uses $\ell+j+1\le k(j+1)$ and $(j+1)^2\le4^j$.

\paragraph{Only fresh transition indices continue.}
Fix a candidate transition index $e$ and a differentiated slot $q$. Let $M_e(\Omega)$ be the indicator that $e$ is absent from the old transition indices. For a new independent ghost $q'$ of the same distribution,
\begin{equation}\label{eq:fresh-ghost-identity}
 \E\bigl[F M_e(I_e(\bm T^q)-\E I_e)\bigr]
 =\E\bigl[F M_e(I_e(\bm T^q)-I_e(\bm T^{q'}))\bigr].
\end{equation}
The mask may depend on $\bm T^q$. Only independence of the new ghost is needed. Introduce a separate fresh gain slot at time $s$ using \eqref{eq:centered-tree-identity} with $\zeta=\alpha$. Split the indicator difference by sign and player. In each branch, the gain root selects a source row, and the indicated tree has at most one outgoing edge from that row. Hence the complete sum over candidate transition indices selects at most one edge.

If the selected edge is fresh, append its $\chi^\alpha$ factor at the differentiated distribution's score type, and put the freshness mask and $\eta_i/\eta_{\max}$ into the bounded observable. Otherwise set the observable to zero and use an arbitrary fixed dummy transition index for that player. Its value is irrelevant on this zero-observable branch. This selection occurs after the fixed-candidate identity \eqref{eq:fresh-ghost-identity}. It neither conditions nor renormalizes a tree distribution.

Each player/sign/distribution branch has coefficient $\delta\eta_{\max}$. There are at most $2n(1+2\ell+2j)$ branches, proving the continuation bound. Two slots and one distinct transition index are added. Old atoms remain unchanged, so their marginal moment bounds remain valid. A bounded sample-dependent mask cannot increase the integrand's $L^2$ norm. At depth $N=|\cI|$, no fresh transition index remains wherever the observable is nonzero.
\end{proof}

\section{Two-scale filter bounds and proof of the prediction bound}\label{app:filters}
This appendix proves the filter bounds and completes \cref{lem:prediction}.

\subsection{Filter bounds}
For a convolution $\mathrm{K}=\sum_{j\ge0}c_j\shift^j$, recall $\kernorm{\mathrm{K}}=\sum_j|c_j|$. Backward shifts contract the finite-prefix squared sum, so
\[
 \sum_{t=1}^T\sourcenorm{(\mathrm{K}[\bm w])^{(t)}}^2
 \le(\kernorm{\mathrm{K}})^2\sum_{t=1}^T\sourcenorm{\bm w^{(t)}}^2.
\]

\begin{lemma}[Two-scale filter bounds]\label{lem:filter-full}
For $0\le a\le\ell$ and $0\le b\le N$,
\begin{equation*}
 \kernorm{\diff^2\filt_\ell^a\filt_N^b}<48,
 \qquad
 \kernorm{\errfilt}<96.
\end{equation*}
Define
\begin{equation*}
 \earlyfilter_p=\res_\ell^\ell\res_N^N\diff^{N+\ell+3-p}
 \quad(0\le p\le\ell+1),
 \qquad
 \latefilter_j=\res_N^N\diff^{N+3-j}
 \quad(0\le j\le N+1).
\end{equation*}
Then $\earlyfilter_0=\errfilt$, $\earlyfilter_p=\earlyfilter_{p+1}\diff$, $\latefilter_j=\latefilter_{j+1}\diff$, and
\begin{equation}\label{eq:filter-families-bounds-full}
 \kernorm{\earlyfilter_{p+1}}<48k^p\quad(0\le p<\ell),
 \qquad
 \earlyfilter_\ell=\res_\ell^\ell\latefilter_0,
 \qquad
 \kernorm{\latefilter_{j+1}}<24(N+1)^j\quad(0\le j\le N).
\end{equation}
\end{lemma}

\begin{proof}
If $F(z)=\sum_{j\ge0}a_jz^j$ is analytic on a neighborhood of the closed unit disk, coefficient Cauchy--Schwarz and normalized circle orthogonality give
\begin{equation}\label{eq:kernel-cauchy}
 \sum_{j\ge0}|a_j|
 \le\sqrt{1+\pi/2}
 \left(\norm F_{L^2(\mathbb T)}^2+\norm{F'}_{L^2(\mathbb T)}^2\right)^{1/2}.
\end{equation}
Indeed, $\sum_{j\ge0}(1+j^2)^{-1}\le1+\int_0^\infty(1+x^2)^{-1}dx=1+\pi/2$. This scalar calculation does not require the sequence output space to be Hilbert.

On $|z|=1$,
\[
 F_r(z)=\frac{1-z}{1-rz/(r+1)},
 \qquad
 |F_r(z)|\le\frac{2(r+1)}{2r+1},
 \qquad
 |F_r(z)|^j<2\quad(0\le j\le r+1).
\]
For $F_{a,b}(z)=(1-z)^2F_\ell(z)^aF_N(z)^b$, differentiation gives
\[
 F_{a,b}'
 =-2(1-z)F_\ell^aF_N^b
 -\frac{a}{\ell+1}F_\ell^{a+1}F_N^b
 -\frac{b}{N+1}F_\ell^aF_N^{b+1}.
\]
Thus $|F_{a,b}|\le16$ and $|F_{a,b}'|\le24$ in the stated ranges. Equation \eqref{eq:kernel-cauchy} gives kernel norm below $48$. One additional difference costs at most two, proving $\kernorm{\errfilt}<96$.

For the single-scale symbol $(1-z)^2F_N(z)^b$, the analogous bounds are $8$ and $10$, giving kernel norm below $24$. Finally,
\begin{equation*}
 \earlyfilter_{p+1}=\res_\ell^p\diff^2\filt_\ell^{\ell-p}\filt_N^N,
 \qquad
 \latefilter_{j+1}=\res_N^j\diff^2\filt_N^{N-j}.
\end{equation*}
The nonnegative resolvent kernels have masses $k^p$ and $(N+1)^j$, respectively, which proves \eqref{eq:filter-families-bounds-full}. These factorizations retain the difference operators. In particular, the bare resolvent $\res_N^N$ has mass $(N+1)^N$ and is not being bounded by a universal constant.
\end{proof}

\subsection{Proof of the prediction bound}

\begin{proof}[Proof of \cref{lem:prediction}]
Starting from a depth-zero gain, let $a_p$ be the total early coefficient mass. By \eqref{eq:early-mass-filtered}, $a_pk^p\le\kappa^p$ for $p\le\ell$. Coefficients are time-independent, so filter every exact early decomposition using $\earlyfilter_p=\earlyfilter_{p+1}\diff$. By \cref{lem:filter-full}, all early stopped contributions sum to at most
\begin{equation*}
 48\sum_{p=0}^{\ell-1}\kappa^p[1+C_E(p+1)gE]
 \le C+CgE.
\end{equation*}
At the handoff, $\earlyfilter_\ell=\res_\ell^\ell\latefilter_0$. The kernel mass of $\res_\ell^\ell$ is $k^\ell$, already paid by $a_\ell k^\ell\le\kappa^\ell$.

For a late depth-zero expression, set
\[
 a_{\mathrm{late}}=2n\delta\eta_{\max}(1+2\ell+2N),
 \qquad
 B_{\mathrm{late}}=1+gE+Y.
\]
Substitute the late closure repeatedly and retain the stopped term at depth $N$. Since $\latefilter_j=\latefilter_{j+1}\diff$,
\begin{equation*}
 \left(\sum_{t=1}^T\sourcenorm{(\latefilter_0[\bm U])^{(t)}}^2\right)^{1/2}
 \le24C_Lk^2B_{\mathrm{late}}\sum_{j=0}^N(8(N+1)a_{\mathrm{late}})^j
 \le25C_Lk^2B_{\mathrm{late}}.
\end{equation*}
Indeed,
\[
 8(N+1)a_{\mathrm{late}}
 \le32n\delta\eta_{\max}(N+\ell+2)^2
 =32\eta_{\max}\le\frac{2}{5^{5/2}}<\frac1{25}.
\]
This is a finite sum of $N+1$ terms, including all boundary impulses. There is no terminal remainder or infinite-depth limit. All filters depend only on current and earlier inputs, so only the same finite-prefix energies appear.

Multiply the late bound by the handoff coefficient $\kappa^\ell$. Since $k^2\kappa^\ell$ is uniformly bounded when $\kappa\le1/16$, adding the unit solver-error norm and collecting the constant and $gE$ terms gives
\[
 E\le C_0+C_1gE+C_Pk^2\kappa^\ell Y,
\]
which is \eqref{eq:prediction}.
\end{proof}

\section{Bracket-and-mix row normalization and implementation}\label{app:solver}
Fix player $i$, one raw score row $\overline{\bm z}\in\R^{m_i}$, and local time $t$. If all coordinates coincide, return the uniform row, which is exact by symmetry. Otherwise initialize
\begin{equation*}
 \mu^- =\min_b\overline z_b-\frac3{m_i\beta_i},
 \qquad
 \mu^+=\max_b\overline z_b.
\end{equation*}
At $\mu^+$, every scalar argument is nonpositive, and the row sum is at most $m_i f_i(0)\le e m_i\beta_i/2<1$. At $\mu^-$, every argument is at least $s=3/(m_i\beta_i)>0$. Since $t_+(s)\ge s$ and $u(s)\le1/s$, Bernoulli's inequality gives
\[
 f_i(s)\ge\beta_i\frac{1+s}{1+s^{-1}+\delta/(4s^2)}>\frac1{m_i}.
\]
Thus the endpoints bracket the unique normalization shift.

Bisect until an exactly normalized midpoint is found or the bracket width is at most
\begin{equation}\label{eq:solver-width-full}
 e_{i,t}=\frac{\eta_i}{m_i(t+1)^2}.
\end{equation}
An exact root is returned immediately. Otherwise put
\begin{equation*}
 y_b^-=f_i(\overline z_b-\mu^-),
 \qquad
 y_b^+=f_i(\overline z_b-\mu^+),
 \qquad
 S^\pm=\sum_by_b^\pm,
 \qquad
 \tau=\frac{1-S^+}{S^--S^+},
\end{equation*}
\begin{equation*}
 \widehat q_b=\tau y_b^-+(1-\tau)y_b^+.
\end{equation*}
Since $S^->1>S^+$, the returned row is strictly positive and exactly normalized.

\begin{lemma}[Effective-score representation]\label{lem:solver-full}
The returned row is the exact normalized row at $\overline{\bm z}+\bm\zeta$ for some vector satisfying $0\le\zeta_b\le\mu^+-\mu^-$ in every coordinate. Across all source rows, this proves \eqref{eq:solver-interface}.
\end{lemma}

\begin{proof}
Monotonicity places $\widehat q_b$ between $f_i(\overline z_b-\mu^+)$ and $f_i(\overline z_b-\mu^-)$. Define, only for the proof,
\[
 \zeta_b=f_i^{-1}(\widehat q_b)-(\overline z_b-\mu^+).
\]
Then $0\le\zeta_b\le\mu^+-\mu^-$ and
\[
 \widehat q_b=f_i(\overline z_b+\zeta_b-\mu^+),
 \qquad
 \sum_b\widehat q_b=1.
\]
Uniqueness of the normalizing shift identifies the returned row with $\bm q_i(\overline{\bm z}+\bm\zeta)$. Exact roots and uniform rows use $\bm\zeta=0$. Set $\bm\varepsilon=\bm\zeta/\eta_i$. For every source row,
\[
 \norm{\bm\varepsilon_{i,a}^{(t)}}_\infty\le\frac1{m_i(t+1)^2}.
\]
Summing over the $m_i$ source rows gives
\[
 \sourcenorm{\bm\varepsilon_i^{(t)}}\le(t+1)^{-2},
 \qquad
 \sum_{t=1}^T\sourcenorm{\bm\varepsilon_i^{(t)}}^2
 \le\sum_{t\ge1}(t+1)^{-4}<1.
\]
\end{proof}

\paragraph{Work, storage, and computation model.}
The operation counts below assume that the public parameters are
supplied as exact real constants. They count exact arithmetic,
comparison, and square-root operations during play. Public
initialization, including the exact logarithms used to define these
parameters, is not included in these per-round counts. No finite-bit
complexity claim is made.
At local round $t$, player $i$ uses
\begin{equation}\label{eq:complexity}
 O\!\left(m_i^2\left[N+\log(nm)\log(nmt)\right]+m_i^3\right)
 \text{ operations and }O\!\left((N+\ell)m_i^2\right)\text{ storage}.
\end{equation}
Since $N\le nm^2$, this is at most
$O\!\left(m^2[nm^2+\log(nm)\log(nmt)]\right)$ operations and $O(nm^4)$ storage.

The public bounds imply $d_i\le C\log(nm)$. Evaluating $f_i$ uses arithmetic, one square root, and $O(\log(nm))$ repeated-squaring operations. The implementation evaluates neither $f_i^{-1}$, the potential, higher derivatives, nor trees. By \eqref{eq:pointwise-bounds}, the raw-score range at round $t$ is at most $2\eta_i(t+96)$, so the initial bracket width divided by \eqref{eq:solver-width-full} is at most
\[
 2m_i(t+96)(t+1)^2+\frac{3(t+1)^2}{\beta_i\eta_i}.
\]
For the default rates, $1/(\beta_i\eta_i)$ is polynomial in $nm$ by $J\le5k$, $A\le N\le nm^2$, and $A_i\ge2$. For the optional rates, \eqref{eq:heterogeneous-rates} gives $\eta_i^{-1}\le N/(2g)$. Hence either allocation requires $O(\log(nmt))$ bisections and yields the normalization term in \eqref{eq:complexity}.

The $N+\ell$ cascade stages use $O((N+\ell)m_i^2)$ work and storage. The exact stationary-distribution system is nonsingular by Appendix~\ref{app:trees-moments}. Gaussian elimination costs $O(m_i^3)$ operations and $O(m_i^2)$ storage. These are per-round counts under the computation model specified above. Inexact stationarity, sampled-action feedback, numerical conditioning, finite-bit complexity, and floating-point behavior are outside the present guarantees.

\section{Anytime fallback and the common-prefix wrapper}\label{app:fallback}
\subsection{An anytime adversarial fallback}\label{app:anytime-fallback}
We define the fallback $\mathsf{FB}$ used in \cref{sec:robustness}. For player $i$, use $A_i=m_i\log m_i$ and doubling epochs with planned lengths $M=1,2,4,\ldots$. Initialize each epoch's row distributions uniformly and set
\[
 \eta_{\mathrm{fb}}=\sqrt{A_i/M}.
\]
Play the stationary distribution $\bm x$ of the strictly positive row-stochastic matrix. After observing $\bm v$, set $g_{a,b}=x_av_b$ and update
\begin{equation*}
 q_{a,b}^+=\frac{q_{a,b}(1+\eta_{\mathrm{fb}}g_{a,b})}
 {1+\eta_{\mathrm{fb}}\ip{\bm q_a}{\bm g_a}}.
\end{equation*}
Rows remain strictly positive and exactly normalized for every rate, because $g_{a,b}\ge0$. This is the Blum--Mansour reduction \citep{blum2007} with a linear multiplicative-weights update.

Comparing the unnormalized weight of a fixed destination with the row-weight sum gives, for an epoch prefix,
\[
 \sum_t\ln(1+\eta_{\mathrm{fb}}g_{a,b}^{(t)})-\ln m_i
 \le\sum_t\ln\left(1+\eta_{\mathrm{fb}}\ip{\bm q_a^{(t)}}{\bm g_a^{(t)}}\right).
\]
For $z\ge0$, $z-z^2/2\le\ln(1+z)\le z$. No upper rate cap is needed. Thus the cumulative row comparison is at most
\[
 \frac{\ln m_i}{\eta_{\mathrm{fb}}}
 +\frac{\eta_{\mathrm{fb}}}{2}\sum_t(g_{a,b}^{(t)})^2.
\]
Sum over source rows with a fixed deviation map. Stationarity identifies the linear comparison with the map's actual gain, while the squared comparator gains satisfy
\[
 \sum_ax_a^2v_{\varphi(a)}^2\le\sum_ax_a^2\le1
\]
per round. A prefix of $\tau\le M$ rounds in the epoch therefore satisfies
\begin{equation*}
 \Reg_{\mathrm{epoch}}(\tau)
 \le\sqrt{A_iM}\left(\ln2+\frac{\tau}{2M}\right).
\end{equation*}

Let the current epoch have planned length $M$ and $1\le\tau\le M$ completed rounds. Total time is $s=M-1+\tau$. All earlier epochs are complete. Put $d=\ln2+1/2$. Their total cost is at most
\[
 d\sqrt{A_i}\frac{\sqrt M-1}{\sqrt2-1}.
\]
With $x=\tau/M$ and $B_M=d(1-M^{-1/2})/(\sqrt2-1)$, divide the total by $\sqrt{A_iM}$. The ratio
\[
 \frac{B_M+\ln2+x/2}{\sqrt{1-M^{-1}+x}}
\]
decreases in $x$: its derivative has the sign of $1-M^{-1}+x/2-(B_M+\ln2)$, which is nonpositive for $M\ge2$ and $x\le1$. The maximum is at $x=1/M$, where $s=M$ and
\[
 B_M+\ln2+\frac1{2M}<2.9+0.7+0.25<4.
\]
For $M=1$, the bound is immediate. Hence $\Reg(s)\le4\sqrt{A_is}$ for every $s\ge0$, proving the fallback guarantee stated in \cref{alg:robust}. The exact stationary-distribution solve costs $O(m_i^3)$ operations and $O(m_i^2)$ storage. Row updates cost $O(m_i^2)$. These bounds use the same exact-real-arithmetic convention as the implementation analysis in Appendix~\ref{app:solver}.

\subsection{Proof of the common-prefix wrapper}\label{app:switching-rule}
Starting immediately with the base dynamics and switching after the first violation of the anytime self-play bound would give $\Reg_i(T)\le B_i+1+4\sqrt{A_iT}$, retaining the base budget additively. The common prefix instead makes this budget absorbable once $B_i+1=O(\sqrt{A_iW})$, while its simultaneous reset makes the subsequent play a fresh self-play trajectory.

\begin{proof}[Proof of \cref{lem:common-prefix-wrapper}]
Swap regret is subadditive over consecutive intervals, since each fixed deviation map's total gain is the sum of its interval gains. Moreover,
$R_{i,\mathrm{base}}(s)\le R_{i,\mathrm{base}}(s-1)+1$, because the one-round gain of every deviation map is at most one.

If a threshold were crossed in self-play, consider the first crossing among all players. Through that round every player has followed the simultaneously initialized fresh base dynamics, contradicting the assumed anytime self-play bound. Hence no threshold-triggered switch occurs.

For an arbitrary payoff sequence, split the horizon into the initial fallback, base, and terminal fallback intervals. If a crossing occurs, the base interval contributes at most $B_i+1$ by the one-round bound. If none occurs, it contributes at most $B_i$. Applying the fallback guarantee stated in \cref{alg:robust} to the two newly initialized fallback intervals and enlarging the terminal interval to length $T-W$ proves \eqref{eq:generic-wrapper-bound}. Under $B_i+1\le\gamma_i\sqrt{A_iW}$, for $T>W$,
\[
 \Reg_i(T)
 \le(4+\gamma_i)\sqrt{A_iW}+4\sqrt{A_i(T-W)}
 \le\sqrt{(4+\gamma_i)^2+16}\,\sqrt{A_iT},
\]
by Cauchy--Schwarz. For $T\le W$, the same fallback guarantee gives the conclusion.
\end{proof}

For the stated extension to a nondecreasing anytime bound, let $b_i(s)$ be nondecreasing. At a first crossing time $\tau$, the one-round bound gives $R_{i,\mathrm{base}}(\tau)\le b_i(\tau)+1$. Hence, under $b_i(s)+1\le\gamma_i\sqrt{A_i(W+s)}$, the switched case is bounded by
\[
 4\sqrt{A_iW}+\gamma_i\sqrt{A_i(W+\tau)}
 +4\sqrt{A_i(T-W-\tau)}
 \le\sqrt{(4+\gamma_i)^2+16}\,\sqrt{A_iT},
\]
where we used $W\le W+\tau$ before Cauchy--Schwarz. Without a crossing, $4\sqrt{A_iW}+b_i(T-W)$ is bounded directly by the same envelope. The same first-crossing argument prevents switching in self-play.

\subsection{Specialization to Algorithm~\ref{alg:base}}\label{app:robustness}
Put
\[
 W_0=\frac{9k^5A}{c^2}.
\]
Then $B_i=\frac43\sqrt{A_iW_0}$ and $W=\lceil W_0\rceil$. Since $c\le1/16$, $k\ge5$, and $A\ge2$, we have $W_0\ge6$, hence $W\le(7/6)W_0$.

By \eqref{eq:generic-bound}, the anytime self-play regret of the base dynamics is at most $3B_i/4$, so \cref{lem:common-prefix-wrapper} prevents switching. The initial fallback satisfies
\[
 4\sqrt{A_iW}\le3\sqrt{7/6}\,B_i<\frac{13}{4}B_i.
\]
Adding the base continuation proves the self-play inequality in \eqref{eq:switching-specialization}, including horizons within the common prefix.

Also $B_i\ge4$, so
\[
 B_i+1\le\frac54B_i\le\frac53\sqrt{A_iW}.
\]
Thus \cref{lem:common-prefix-wrapper} applies with $\gamma_i=5/3$, and
\[
 \sqrt{(4+5/3)^2+16}=\sqrt{433/9}<7.
\]
This proves the adversarial part of \cref{thm:robust}. The preceding paragraph proves its self-play part.

Maintaining \eqref{eq:switching-monitor} costs $O(m_i^2)$ work and storage per round and does not change the stated asymptotic implementation bounds. All states are discarded at phase transitions, and the threshold is evaluated only after the crossing round's feedback has been observed. Every round, including the common prefix and crossing round, has been counted.

\end{document}